\documentclass[conference]{IEEEtran}

\newcommand{\myparatight}[1]{\smallskip\noindent{\bf {#1}:}~}

\usepackage{amsthm}
\usepackage{amsmath}

\usepackage{float}
\usepackage{graphics}
\usepackage{graphicx}

\usepackage{algorithm}
\usepackage{algorithmicx}
\usepackage{algpseudocode}

\usepackage{bm}
\usepackage{color}
\usepackage{multirow}
\usepackage{makecell}
\usepackage{balance}
\usepackage{pdfx}
\usepackage{subfig}

\usepackage{amssymb}
\usepackage{xspace}

\usepackage{placeins}

\usepackage[misc,geometry]{ifsym}

\newtheorem{assumption}{Assumption}
\newtheorem{thm}{Theorem}
\newtheorem{lem}{Lemma}
\newtheorem*{remark}{Remark}

\definecolor{myurlcolor}{rgb}{0.1, 0.2, 0.8}

\usepackage{hyperref}
\hypersetup{
	colorlinks=true,
	linkcolor=blue,
	filecolor=magenta,      
        urlcolor=myurlcolor,
        citecolor=blue,
	pdftitle={Overleaf Example},
	pdfpagemode=FullScreen,
}

\usepackage{color, colortbl}
\definecolor{greyL}{RGB}{230,248,255}

\newcommand{\alg}{\textsf{FoundationFL}\xspace}

\algnewcommand\algorithmicforpara{\textbf{for}}
\algnewcommand\algorithmicdoinparallel{\textbf{do in parallel}}
\algdef{S}[FOR]{ForParallel}[1]{\algorithmicforpara\ #1\ \algorithmicdoinparallel}

\begin{document}

\title{Do We Really Need to Design New Byzantine-robust Aggregation Rules?}

\author{\IEEEauthorblockN{
Minghong Fang\IEEEauthorrefmark{1}\textsuperscript{, \Letter}\thanks{\Letter \ Minghong Fang is the corresponding author.},
Seyedsina Nabavirazavi\IEEEauthorrefmark{2},
Zhuqing Liu\IEEEauthorrefmark{5},\\
Wei Sun\IEEEauthorrefmark{4}, 
Sundararaja Sitharama Iyengar\IEEEauthorrefmark{2},
Haibo Yang\IEEEauthorrefmark{3}}
\IEEEauthorblockA{\IEEEauthorrefmark{1}University of Louisville, \IEEEauthorrefmark{2}Florida International University, \IEEEauthorrefmark{5}University of North Texas,} \newline\IEEEauthorrefmark{4}Wichita State University,\IEEEauthorrefmark{3}Rochester Institute of Technology}

\IEEEoverridecommandlockouts
\makeatletter\def\@IEEEpubidpullup{6.5\baselineskip}\makeatother
\IEEEpubid{\parbox{\columnwidth}{
		Network and Distributed System Security (NDSS) Symposium 2025\\
		24-28 February 2025, San Diego, CA, USA\\
		ISBN 979-8-9894372-8-3\\
		https://dx.doi.org/10.14722/ndss.2025.241796\\
		www.ndss-symposium.org
}
\hspace{\columnsep}\makebox[\columnwidth]{}}

% make the title area
\maketitle

% !TEX root = main.tex

\begin{abstract}
Federated learning (FL) allows multiple clients to collaboratively train a global machine learning model through a server, without exchanging their private training data. However, the decentralized aspect of FL makes it susceptible to poisoning attacks, where malicious clients can manipulate the global model by sending altered local model updates. To counter these attacks, a variety of aggregation rules designed to be resilient to Byzantine failures have been introduced. Nonetheless, these methods can still be vulnerable to sophisticated attacks or depend on unrealistic assumptions about the server. In this paper, we demonstrate that there is no need to design new Byzantine-robust aggregation rules; instead, FL can be secured by enhancing the robustness of well-established aggregation rules. To this end, we present FoundationFL, a novel defense mechanism against poisoning attacks. FoundationFL involves the server generating synthetic updates after receiving local model updates from clients. It then applies existing Byzantine-robust foundational aggregation rules, such as Trimmed-mean or Median, to combine clients' model updates with the synthetic ones. We theoretically establish the convergence performance of FoundationFL under Byzantine settings. Comprehensive experiments across several real-world datasets validate the efficiency of our FoundationFL method.
\end{abstract}
% !TEX root = main.tex

\section{Introduction} \label{sec:intro}
In recent years, federated learning (FL) has emerged as a promising approach to distributed learning~\cite{mcmahan2017communication}. It allows multiple clients to collaboratively train a global machine learning model (called \emph{global model}) under the coordination of a central server, all while respecting the privacy of clients' sensitive training data. Essentially, in each training round, the server distributes the current global model to all clients or a subset of them. Each selected client refines its local machine learning model (called \emph{local model}) by using this global model and its own local training data. Subsequently, the client sends its local model update back to the server. Upon receiving updates from clients, the server aggregates these updates into a global model update, which it then integrates to further update the global model.
FL has been implemented across a range of practical tasks and applications, including credit risk assessment~\cite{webank}, predictive text input~\cite{gboard}, and speech recognition~\cite{paulik2021federated}.

However, the decentralized nature of FL poses distinct challenges, with one of the most significant being its susceptibility to poisoning attacks~\cite{bagdasaryan2020backdoor,baruch2019little,bhagoji2019analyzing,blanchard2017machine,fang2020local,shejwalkar2021manipulating,tolpegin2020data,xie2019dba,yin2024poisoning,zhang2024poisoning}. In these attacks, malicious clients, under the control of an attacker, attempt to compromise the integrity of the global model. They do this by manipulating their local training data or sending carefully crafted updates directly to the server. 
Depending on the attacker's objectives, poisoning attacks can be categorized as untargeted~\cite{blanchard2017machine,fang2020local,shejwalkar2021manipulating,tolpegin2020data} or targeted~\cite{bagdasaryan2020backdoor,baruch2019little,cao2020fltrust}. 
In untargeted attacks, the goal is to degrade the overall performance of the global model.
In contrast, targeted attacks aim to induce incorrect predictions specifically on testing inputs chosen by the attacker, while leaving predictions for other inputs unaffected.
It has been demonstrated that a single malicious client is sufficient to successfully compromise an FL system that uses a straightforward average aggregation strategy. In this scenario, one malicious client can arbitrarily manipulate the final aggregated result~\cite{blanchard2017machine}.

In response to poisoning attacks, researchers have shifted from using a straightforward average aggregation rule to developing robust aggregation methods to increase the resilience of FL. These include the Trimmed-mean and Median aggregation rules introduced by~\cite{yin2018byzantine}, which are termed \emph{Byzantine-robust foundational aggregation rules}. Trimmed-mean is a coordinate-wise aggregation protocol that removes some of the largest and smallest extreme values for each dimension before averaging the remaining values. In contrast, the Median method calculates the coordinate-wise median of the clients' local model updates. Despite their robustness, recent research indicates that these rules are still susceptible to advanced poisoning attacks~\cite{fang2020local,shejwalkar2021manipulating}. To combat these vulnerabilities, newer and more complex aggregation rules have been developed~\cite{blanchard2017machine,cao2020fltrust,fereidooni2023freqfed,munoz2019byzantine,nguyen2022flame,pan2020justinian,park2021sageflow,rieger2022deepsight,wang2022flare,xie2019zeno}. These, however, either depend heavily on the FL system having access to a clean dataset~\cite{cao2020fltrust,pan2020justinian,park2021sageflow,wang2022flare,xie2019zeno} or continue to be prone to sophisticated attacks~\cite{shejwalkar2022back,xie2024poisonedfl}. This ongoing escalation between attackers and defenders often results in a costly and potentially unsustainable arms race. This situation raises an essential research question: \emph{\textbf{Do we really need to design new Byzantine-robust aggregation rules?}}

In this paper, we find a promising answer to the aforementioned question. Instead of developing complex new Byzantine-robust aggregation protocols, we aim to enhance the robustness of FL systems by employing well-established Byzantine-robust aggregation methods like Trimmed-mean and Median. The unique aspect of FL is the non-identical and non-independent (Non-IID) distribution of training data among clients, which introduces substantial diversity. This diversity allows malicious clients to manipulate their local model updates to undermine the FL system, while remaining distinct from benign clients. To mitigate these issues of heterogeneity, in our proposed \alg framework, the server introduces synthetic updates in each global training round. The primary challenge then becomes determining these synthetic updates. To tackle this, the server calculates a closeness score for each client in every round to assess how their local model updates align with the most extreme updates. The server then selects the client's local model update that deviates the most from these extreme updates. This chosen update is used as the basis for the synthetic updates, which are then mixed with the clients’ local model updates. Ultimately, the server applies Byzantine-robust foundational aggregation rules like Trimmed-mean or Median to combine both the clients' local and the synthetic updates.

We offer theoretical guarantees for our \alg framework under poisoning attacks. Specifically, we provide a theoretical demonstration that the global model learnt through our \alg, even when subjected to poisoning attacks, converges with high probability to the optimal global model that would be obtained in the absence of attacks, given certain mild assumptions. We conduct comprehensive evaluations of our \alg across 6 datasets spanning various domains, against 12 different poisoning attacks, and comparing with 10 FL aggregation rules. Our findings show that our proposed \alg significantly surpasses current Byzantine-robust FL methods in performance.

We summarize our main contributions in this paper as follows:

\begin{list}{\labelitemi}{\leftmargin=1em \itemindent=-0.08em \itemsep=.2em}
\item
We introduce \alg, a robust aggregation framework designed to combat poisoning attacks within FL environments.

\item
We demonstrate theoretically that \alg remains resilient to poisoning attacks under commonly accepted assumptions within the Byzantine-robust FL community.

\item
Our thorough experiments across diverse benchmark datasets, various poisoning attack scenarios, and practical FL setups validate the effectiveness of our \alg framework.

\end{list}	

% !TEX root = main.tex

\section{Background and Related Work} \label{sec:related}

\myparatight{Notations} 
In this paper, the $\ell_2$-norm is denoted by $\left\| \cdot \right\|$. Furthermore, for any natural number $n$, the set $\left\{1, \ldots, n\right\}$ is represented as $[n]$. Additionally, matrices and vectors are specified using bold typeface.

\subsection{Federated Learning}

Federated learning (FL) usually includes a server and $n$ participating clients that work together to train a global model without exchanging private data. 
We denote the local training dataset of client $i$ as $D_i$, where $i \in [n]$. The goal of FL is to minimize the following global
objective function:
\begin{align}
    \label{fl_goal}
   {L}(\bm{\theta})= \frac{1}{n}\sum\limits_{i=1}^n {L}_i(\bm{\theta}),
\end{align}
where $\bm{\theta} \in \mathbb{R}^d$ is the model parameter, $d$ is the dimension of $\bm{\theta}$,
${L}_i(\bm{\theta}) = \frac{1}{|D_i|}\sum\nolimits_{z \in D_i} {l}(\bm{\theta}, z)$ is the local training objective (empirical loss) of client $i$, $|D_i|$ denotes the count of training example for the client $i$.

FL solves the above Problem~(\ref{fl_goal}) through an iterative process.
Specifically, in each global training round $t$, the following three specific steps are executed:

% \begin{itemize}
\begin{list}{\labelitemi}{\leftmargin=1em \itemindent=-0.08em \itemsep=.2em}

\item \textbf{Step I: Global model synchronization.}
The server distributes the current global model $\bm{\theta}^t$ to all clients or a selected group of them.

\item \textbf{Step II: Local model updating.}
Each client $i \in [n]$ fine-tunes its local model $\bm{\theta}_i^t$ by leveraging the current global model $\bm{\theta}^t$ and its local training dataset $D_i$.
Then, client $i$ transmits its local model update $\bm{g}_i^t = \bm{\theta}_i^t - \bm{\theta}^t$ to the server.

\item \textbf{Step III: Global model updating.}
After collecting model updates from the clients, the server employs the aggregation rule, denoted as $\text{Agg}$, to merge these updates to get a global model update. The global model is then updated as $\bm{\theta}^{t+1} = \bm{\theta}^{t} + \eta \cdot \text{Agg} \{\bm{g}_i^t: i \in [n]\}$, where $\eta$ is the learning rate.

\end{list}

FL repeats the aforementioned three steps across multiple global training rounds until a specified convergence condition is satisfied.
It's worth mentioning that various FL methods often employ distinct aggregation protocols~\cite{mcmahan2017communication,blanchard2017machine,yin2018byzantine}. For instance, in the case of the FedAvg~\cite{mcmahan2017communication} aggregation rule, the server combines the model updates as $\text{Agg} \{\bm{g}_i^t: i \in [n]\} = \frac{1}{n}\sum\limits_{i=1}^n \bm{g}_i^t$.

\subsection{Poisoning Attacks to FL}

Although FL's decentralized architecture offers privacy benefits, it also renders it susceptible to poisoning attacks.
In a FL environment, malicious clients may attempt to manipulate the global model's performance by interfering with the training process. This could involve corrupting their local training data (known as \emph{data poisoning attacks}~\cite{tolpegin2020data}) or tampering with their local model updates (known as \emph{local model poisoning attacks}~\cite{blanchard2017machine,fang2020local,shejwalkar2021manipulating}). The ultimate goal of these malicious clients is to degrade the global model's performance. For instance, the final learnt global model exhibits reduced accuracy when classifying testing examples.
For instance, in a label flipping attack~\cite{tolpegin2020data}, malicious clients flip the labels associated with their training data while leaving the underlying features unchanged. 
In the Trim attack~\cite{fang2020local}, malicious clients strategically manipulate their local model updates before sending them to the central server. This manipulation aims to exploit the vulnerabilities of Trimmed-mean~\cite{yin2018byzantine} or Median~\cite{yin2018byzantine} aggregation rule employed by the server, causing a significant deviation in the global model updates after the attack compared to the updates before the attack.

\subsection{Byzantine-robust Aggregation Rules}
In standard FL, the server combines the local model updates it receives by computing their averages~\cite{mcmahan2017communication}. Nevertheless, recent research~\cite{blanchard2017machine} has revealed that this aggregation method, based on averaging, is highly susceptible to Byzantine attacks. In such attacks, a single malicious client has the ability to manipulate the final aggregated outcome in arbitrary ways.
In order to safeguard FL against poisoning attacks, several Byzantine-robust foundational aggregation rules, such as Krum~\cite{blanchard2017machine}, Trimmed-mean~\cite{yin2018byzantine} and Median~\cite{yin2018byzantine}, have been suggested.
For instance, with the Krum~\cite{blanchard2017machine} aggregation rule, once the server receives $n$ local model updates from clients, it selects and outputs the local model update that has the smallest total distance to its $n-f-2$ closest neighbors, where $n$ is the total number of clients, $f$ is the number of malicious clients.
Trimmed-mean~\cite{yin2018byzantine} aggregation rule operates on a per-coordinate basis. Specifically, for every dimension, the server starts by eliminating the largest $c$ elements and the smallest $c$ elements, and subsequently computes the average of the remaining values, where $c$ is the trim parameter.
Median~\cite{yin2018byzantine} is another coordinate-wise aggregation method. 
In this approach, the server combines the received local model updates by determining the median value for each dimension of the updates.
In recent years, several Byzantine-robust advanced aggregation rules have been proposed~\cite{cao2020fltrust,nguyen2022flame,pan2020justinian,park2021sageflow,wang2022flare,xie2019zeno,ChenPOMACS17,fang2022aflguard,fang2024byzantine,kumari2023baybfed,Mhamdi18,xu2024robust}.
For instance, in Bulyan~\cite{Mhamdi18} aggregation rule, the server first leverage Krum~\cite{blanchard2017machine} aggregation rule to select a select of local model updates, then take the Median~\cite{yin2018byzantine} of these selected updates.
The authors in~\cite{cao2020fltrust,pan2020justinian,park2021sageflow,wang2022flare,xie2019zeno} assume that the server possesses a clean validation dataset, sourced from the same distribution as the overall training dataset. Utilizing this validation dataset, the server calculates a benchmark model update. This benchmark is then applied to determine if the local model updates received are either benign or malicious.

\myparatight{Limitations of existing robust aggregation rules}Firstly, current Byzantine-resistant aggregation methods are not entirely secure, as they remain susceptible to sophisticated poisoning attacks~\cite{fang2020local,shejwalkar2021manipulating}.
Secondly, numerous robust aggregation rules rely on strong assumptions regarding the server's capabilities, such as possessing a separate validation data~\cite{cao2020fltrust,pan2020justinian,park2021sageflow,wang2022flare,xie2019zeno}. However, this assumption is often impractical as it is challenging for the server to accurately know the distribution of clients' local training data. 
% Moreover, the possession of validation data by the server could infringe upon privacy concerns, which contradicts the foundational principles of FL's design.
%
Moreover, the possession of validation data by the server could infringe upon privacy concerns, contradicting the core principles of FL's design.

% !TEX root = main.tex

\section{Threat Model} \label{sec:problem}

\myparatight{Attacker’s goal and knowledge}
In our paper, we consider the attack model as described in~\cite{fang2020local,shejwalkar2021manipulating,cao2020fltrust}. Specifically, the attacker manipulates certain malicious clients, which may either be fake clients injected by the attacker or benign clients compromised by the attacker. These controlled malicious clients send meticulously crafted updates of local models to the server to achieve the attacker's objectives. For instance, in untargeted attacks, the goal is to corrupt the resulting global model such that it incorrectly classifies a significant portion of test examples without distinction. In targeted attacks, the objective is to manipulate the global model so that it predicts specific instances chosen by the attacker to match predetermined labels.
We consider the worst-case but realistic attack scenario where the attacker knows the aggregation rule used by the server, and all clients' local model updates. For example, the server may public its aggregation protocol, and the attacker may eavesdrop the communication link in order to get access to local model updates on the benign clients. We note that in our proposed \alg framework, the server generates some synthetic updates. These generated synthetic updates are not access to the attacker, since the server in FL is secure, it hard and even impossible for the attacker to compromise the server in order to have access to these synthetic updates.

\myparatight{Defender’s goal and knowledge} 
Our goal is to develop an effective Byzantine-robust method that accomplishes the following three objectives:
\begin{list}{\labelitemi}{\leftmargin=1em \itemindent=-0.08em \itemsep=.2em}

\item \textbf{Competitive performance:} The proposed defense scheme for FL should also perform effectively in non-adversarial environments. Specifically, in the absence of malicious clients, the model trained using our algorithm should achieve a testing error rate comparable to that of averaging-based aggregation, which is known to deliver state-of-the-art performance in non-adversarial FL settings.

\item \textbf{Byzantine robustness:} The proposed method should demonstrate robustness against Byzantine attacks, both empirically and theoretically.

\item \textbf{Efficiency:} The proposed algorithm should not increase the communication costs between the server and clients, nor should it lead to significant computational demands on the server side.

\end{list}

Regarding the defender's knowledge, the defender (server) is unaware of the attacker's methods of conducting attacks. Additionally, the defender does not have knowledge about the distribution of the clients' local training data.
Note that, 
following~\cite{fang2020local,shejwalkar2021manipulating,shejwalkar2022back}, we assume in our threat model that the majority of clients are benign.

% !TEX root = main.tex

%\vspace{-.1in}
\section{Our Method} 
\label{our_method}

\subsection{Overview} 
In this section, we provide a formal demonstration that the server can enhance the robustness of the FL system by using established Byzantine-robust foundational aggregation rules. This indicates that creating new, complex Byzantine-robust aggregation protocols is unnecessary. In our framework, the server takes a proactive approach by generating synthetic updates upon receiving local model updates from clients. Following this, the server employs existing Byzantine-robust foundational aggregation protocols, such as Trimmed-mean or Median, to combine the local model updates from clients with the generated synthetic updates.

\subsection{\alg}

\myparatight{Trimmed-mean~\cite{yin2018byzantine}}Let $\bm{g}^t$ represent the global model update at the training round $t$, where $\bm{g}^t = \text{Agg} \{\bm{g}_i^t: i \in [n]\}$. The $k$-th component of vector $\bm{g}^t$, denoted $\bm{g}^t[k]$ for $k \in [d]$, is calculated using the Trimmed-mean aggregation method. In this method, for each dimension, the largest $c$ and smallest $c$ values are discarded, and the mean of the remaining values is computed. Specifically, $\bm{g}^t[k]$ is obtained by $\bm{g}^t[k] = \text{Trimmed-mean}\{\bm{g}_1^t[k], \ldots, \bm{g}_n^t[k]\}$.

\myparatight{Median~\cite{yin2018byzantine}}The Median is another rule for aggregation that also operates on each coordinate individually. For every dimension, the server calculates the median of the values from the clients' local model updates. Specifically, for the $k$-th dimension, $\bm{g}^t[k]$ is determined by $\bm{g}^t[k] = \text{Median}\{\bm{g}_1^t[k], \ldots, \bm{g}_n^t[k]\}$.

In FL, a distinctive feature is the non-identical and non-independent (Non-IID) distribution of clients' training data. This diversity across different clients means that the training data is not uniformly distributed, making it significantly varied. As clients train their local models using the current global model and their unique training data, even benign local model updates exhibit notable differences. As a result, this heterogeneity can mask the activities of malicious clients, who may take advantage of these differences to manipulate their local model updates and launch poisoning attacks without detection. These malicious clients craft their updates that, while aimed at undermining the global model, appear as normal within the range of local model updates from benign clients. This subtle manipulation makes it exceedingly challenge to detect and neutralize such threats within the FL system. Moreover, this inherent heterogeneity underscores why existing Byzantine-robust foundational aggregation rules, such as Trimmed-mean~\cite{yin2018byzantine} and Median~\cite{yin2018byzantine} (described further above), remain vulnerable to Byzantine attacks. This vulnerability has been illustrated in~\cite{fang2020local,shejwalkar2021manipulating}, highlighting the ongoing challenges in securing FL systems against sophisticated poisoning attacks.

To address the issue of data heterogeneity in FL and to enhance system robustness, our approach focuses on leveraging existing Byzantine-robust foundational aggregation rules, rather than developing new ones. The core concept of our proposed \alg framework involves a proactive step by the server: after receiving local model updates from clients during each global training round, the server generates additional synthetic model updates. These synthetic updates, when combined with the clients' local model updates, are then aggregated using established Byzantine-robust aggregation rules like Trimmed-mean or Median. The fundamental advantage of our proposed method is that by introducing synthetic model updates, we create a more homogeneous set of updates—where the augmented model updates (comprising both clients' local model updates and the synthetic updates) exhibit much lower variance compared to the original solely client-sourced updates. This reduction in variance across the updates makes the aggregated model less susceptible to outliers and potential poisoning attacks. By feeding these more uniform updates into proven robust aggregation mechanisms, we significantly enhance the system's ability to thwart malicious interventions, thereby increasing the overall security and reliability of the FL system. This strategy leverages the strengths of existing robust aggregation frameworks while effectively countering the challenges posed by data heterogeneity in federated environments.
In what follows, we demonstrate how to construct the synthetic updates.

In each global training round $ t $, assuming the server generates $ m $ synthetic updates, represented as $ \{\bar{\bm{g}}_1^t, \ldots, \bar{\bm{g}}_m^t \} $. 
The central challenge lies in determining these \( m \) updates effectively.
% the key challenge boils down to how to determine these $ m $ updates effectively.
%
Remember that Trimmed-mean and Median are robust statistical methods designed to eliminate outliers (extreme values) from each dimension of clients' local model updates.
Motivated by this observation, the server first identifies a client's local model that deviates the most from the extreme updates, and augments clients' local model updates by incorporating multiple copies of this selected update.
Define $\bm{g}_{\text{max}}^t \in \mathbb{R}^d$ and $\bm{g}_{\text{min}}^t \in \mathbb{R}^d$ as the vectors representing the largest and smallest updates across all dimensions, respectively. Specifically, $\bm{g}_{\text{max}}^t[k]$ is the maximum value of the set $\{\bm{g}_1^t[k], \bm{g}_2^t[k], \ldots, \bm{g}_n^t[k]\}$, and $\bm{g}_{\text{min}}^t[k]$ is the minimum value of the same set for the dimension $k$. 
The values of $\bm{g}_{\text{max}}^t[k]$ and $\bm{g}_{\text{min}}^t[k]$ are determined as:
\begin{align}
\bm{g}_{\text{max}}^t[k] &= \max \{\bm{g}_1^t[k], \ldots, \bm{g}_n^t[k]\}, k \in [d] \\
\bm{g}_{\text{min}}^t[k] &= \min \{\bm{g}_1^t[k], \ldots, \bm{g}_n^t[k]\}, k \in [d].
\end{align}

Upon deriving $\bm{g}_{\text{max}}^t$ and $\bm{g}_{\text{min}}^t$, the server assigns a score $s_i^t$ to each client $i \in [n]$. This score quantifies how closely each client $i$'s local model update $\bm{g}_i^t$ aligns with the extreme updates, namely $\bm{g}_{\text{max}}^t$ and $\bm{g}_{\text{min}}^t$. A higher $s_i^t$ indicates a greater likelihood that the update $\bm{g}_i^t$ is malicious. In our proposed \alg, the server calculates $s_i^t$ by taking the lesser of the distances between $\bm{g}_i^t$ and the vectors $\bm{g}_{\text{max}}^t$ and $\bm{g}_{\text{min}}^t$, as shown below:
\begin{align}
\label{s_i_t}
s_i^t = \min\{ \|\bm{g}_i^t - \bm{g}_{\text{max}}^t\|, \|\bm{g}_i^t - \bm{g}_{\text{min}}^t\| \}.
\end{align}

The server then selects one local model update from the set $\{\bm{g}_1^t, \ldots, \bm{g}_n^t\}$ that exhibits the greatest deviation from the extreme updates $\bm{g}_{\text{max}}^t$ and $\bm{g}_{\text{min}}^t$. Let $i_* \in [n]$ be defined as the client with the largest score, such that $s_{i_*}^t \geq s_i^t$ for all $i \in [n]$. 
Formally, this selection criterion is defined as:
\begin{align}
\label{i_star}
i_* = \operatorname*{argmax}_{i \in [n]} s_i^t.
\end{align}

Following that, the server designates each synthetic update as $\bar{\bm{g}}_j^t=\bm{g}_{i_*}^t$ for all $j \in [m]$. It then supplements the clients' local model updates with these $m$ synthetic updates. To derive the global model update, the server applies Byzantine-robust foundational aggregation protocols. The resulting global model update, denoted by $\hat{\bm{g}}^t$, is computed as follows: if the Trimmed-mean protocol is employed for aggregation, each dimension $k \in [d]$ is updated according to:
\begin{align}
\label{trim_update}
\hat{\bm{g}}^t[k] &= \text{Trimmed-mean}\{\bm{g}_1^t[k], \ldots, \bm{g}_n^t[k], \bar{\bm{g}}_1^t[k],\ldots,\bar{\bm{g}}_m^t[k]\} \\
&  \text{s.t. }  \medspace \bar{\bm{g}}_1^t[k] = \cdots = \bar{\bm{g}}_m^t[k] = \bm{g}_{i_*}^t[k].   \nonumber
\end{align}

Similarly, if the Median protocol is used, the global model update for each dimension $k \in [d]$ is determined by:
\begin{align}
\label{median_update}
\hat{\bm{g}}^t[k] &= \text{Median}\{\bm{g}_1^t[k], \ldots, \bm{g}_n^t[k], \bar{\bm{g}}_1^t[k],\ldots,\bar{\bm{g}}_m^t[k]\} \\
&  \text{s.t. }  \medspace 
 \bar{\bm{g}}_1^t[k] = \cdots = \bar{\bm{g}}_m^t[k] = \bm{g}_{i_*}^t[k].   \nonumber
\end{align}

Finally, the server updates the global model with $\bm{\theta}^{t+1} = \bm{\theta}^{t} + \eta \cdot \hat{\bm{g}}^t$. It is important to note that the server does not update the global model directly using the selected local model update $\bm{g}_{i_*}^t$. Although $\bm{g}_{i_*}^t$ shows the greatest deviation from the extreme updates, it may still contain extreme values in certain dimensions. Thus, the server employs coordinate-wise robust aggregation methods like the Trimmed-mean or Median to mitigate the influence of outliers in each dimension.
Algorithm~\ref{training_procedure_our} shows the pseudocode of \alg framework.

\begin{algorithm}[t]
    \caption{Training procedure of \alg.}
    \label{training_procedure_our}
    \begin{algorithmic}[1]
    \renewcommand{\algorithmicrequire}{\textbf{Input:}}
    \renewcommand{\algorithmicensure}{\textbf{Output:}}
     \Require The $n$ clients with local training dataset; number of global training rounds $T$; number of synthetic updates $m$; learning rate $ \eta$; Byzantine-robust foundational aggregation rule $\text{Agg}$.
     \Ensure Global model $\bm{\theta}^T$. 
       \State Initialize $\bm{\theta}^0$.
        \For {$t = 0, 1, \cdots, T-1$}
         \State  // Step I: Global model synchronization.
        \State The server send the global model $\bm{\theta}^t$ to all clients.
         \State  // Step II: Local model updating.
        \For {each client $i\in [n]$ in parallel}
        \State Client $i$ fine-tunes its local model and sends the local model update $\bm{g}_i^t$ to the server. 
        \EndFor
        \State  // Step III: Global model updating.
        \State The server compute the score $s_i^t $ for $i \in [n]$ based on Eq.~(\ref{s_i_t}).
         \State The server chooses client \( i_* \) with the highest score as defined in Eq.~(\ref{i_star}), and generates each synthetic update as \( \bar{\bm{g}}_j^t = \bm{g}_{i_*}^t \) for all \( j \in [m] \).
         \State The server computes the global model update $\hat{\bm{g}}^t$ based on Eq.~(\ref{trim_update}) if using the Trimmed-mean aggregation rule, or according to Eq.~(\ref{median_update}) if using the Median aggregation rule.
         \State The server updates the global model as $\bm{\theta}^{t+1} = \bm{\theta}^{t} + \eta \cdot \hat{\bm{g}}^t$.
        \EndFor
    \end{algorithmic}
\end{algorithm}

% !TEX root = main.tex

\section{Theoretical Analysis} 
\label{theoretical_analysis}

In this section, we present a convergence analysis of our proposed \alg framework. In our theoretical proof, we assume that the server generates synthetic model updates using a clean dataset \( D_0 \). Both \( D_0 \) and the collective training data \( D = \bigcup_{i=1}^{n} D_{i} \) from clients are presumed to be drawn from the same distribution. This assumption is strictly for theoretical analysis purposes. As demonstrated in Section~\ref{our_method}, the server generates synthetic model updates based exclusively on the received model updates from clients. 
Define $Q$ as $Q=\max\{|D_0|, |D_1|,\cdots,|D_{n}|\}$.
%%%%%%%%%%%%%
Additionally, we adopt the standard assumptions prevalent in the FL literature~\cite{yin2018byzantine,ChenPOMACS17,chu2022securing,karimireddy2020byzantine}.

\begin{assumption}
\label{assumption_1}
The loss function $\mathcal{L}(\bm{\theta})$ is $\mu$-strongly convex. Let $\Theta$ represent the parameter space.
For any  $\bm{\theta}_1, \bm{\theta}_2 \in \Theta$, the following inequality holds:
\begin{gather*}
	L(\bm{\theta}_1) + \left\langle {\nabla L(\bm{\theta}_1),\bm{\theta}_2 - \bm{\theta}_1} \right\rangle + \frac{\mu}{2}{\left\| {\bm{\theta}_2 - \bm{\theta}_1} \right\|^2} \le
	L(\bm{\theta}_2). \nonumber
\end{gather*}
\end{assumption}

\begin{assumption}
\label{assumption_2}
The loss functions are $\lambda$-smooth. For any  $\bm{\theta}_1, \bm{\theta}_2 \in \Theta$, the following inequalities are satisfied:
\begin{gather*}
\left\|  \nabla L(\bm{\theta}_1) -  \nabla L(\bm{\theta}_2) \right\| \le \lambda  \left\| \bm{\theta}_1- \bm{\theta}_2  \right\|, \nonumber \\
\left\|  \nabla l(\bm{\theta}_1, D) -  \nabla l(\bm{\theta}_2, D) \right\| \le \lambda  \left\| \bm{\theta}_1- \bm{\theta}_2  \right\|.  \nonumber
\end{gather*}
\end{assumption}

\begin{assumption}
\label{assumption_3}
The diameter of the parameter space is limited. Specifically, for any $\bm{\theta}_1, \bm{\theta}_2 \in \Theta$, it can be stated that:
\begin{gather*}
 \left\|  \bm{\theta}_1 - \bm{\theta}_2 \right\| \le \varpi. \nonumber
\end{gather*}
\end{assumption}

\begin{assumption}
\label{assumption_4}
The expected squared norm of gradient is bounded by \(\zeta\), and the variance of gradient is bounded by \(\sigma^2\).
Specifically, for any \(\bm{\theta} \in \Theta\), the following inequalities hold:
\begin{gather*}
\mathbb{E} [ \left\|  \nabla l(\bm{\theta}, D)   \right\|^2 ]  \le  \zeta,  \nonumber \\
\mathbb{E} [  \left\|  \nabla l(\bm{\theta}, D)   -   \mathbb{E}[\nabla l(\bm{\theta}, D)] \right\|^2 ]   \le   \sigma^2.  \nonumber 
\end{gather*}
\end{assumption}

\begin{assumption}
\label{assumption_5}
For a given dimension \(k \in [d]\), let \(\partial_k l(\bm{\theta}, D)\) denote the partial derivative of \(l(\bm{\theta}, D)\) with respect to \(\bm{\theta}[k]\), where \(\bm{\theta}[k]\) represents the \(k\)-th dimension of the model parameter \(\bm{\theta}\). We assume that \(\partial_k l(\bm{\theta}, D)\) is \(\rho\)-sub-exponential for any \(k \in [d]\).
\end{assumption}

Based on the above assumptions, we present the theoretical results of our proposed \alg framework.

\begin{thm}
\label{Theorem1}
Assuming that Assumptions~\ref{assumption_1}-\ref{assumption_3} and Assumption~\ref{assumption_5} are valid and the client's learning rate is $\alpha=\frac{1}{\lambda}$, our proposed \alg framework uses the Trimmed-mean aggregation rule to combine both generated synthetic model updates and model updates from clients. 
Given $\upsilon > 0$, if $\beta=\frac{f}{n+m}$ and the trim parameter $c$ meet the criteria $\beta \le \frac{c}{n+m} \le \frac{1}{2} - \upsilon$, where $f$ is the number of malicious clients.
Then after $T$ rounds of global training, the probability of achieving the following result is at least $1-\frac{4d}{(1+ (n+m) \lambda \varpi Q)^d}$:
\begin{align}
 \left\| \bm{\theta}^T - \bm{\theta}^* \right\| \le 
  (1-\frac{\mu}{\mu+\lambda})^T   \left\| \bm{\theta}^0 - \bm{\theta}^* \right\| + \frac{2 B_1}{\mu}, \nonumber
  \end{align}
where $\bm{\theta}^T$ is the global model at training round $T$,  $\bm{\theta}^0$ is the initial global model, $\bm{\theta}^*$ is the optimal model under no attack, 
$B_1=\mathcal{O}((\frac{\rho c d}{\upsilon (n+m) \sqrt{Q}} + \frac{\rho d}{\upsilon \sqrt{(n+m) Q}}) \sqrt{ \log((n+m) \lambda \varpi Q)} ) $.
\end{thm}

\begin{proof}
   The proof is relegated to Appendix~\ref{sec:appendix_1}.
\end{proof}

\begin{thm}
\label{Theorem2}
Under the assumptions that Assumptions~\ref{assumption_1}-\ref{assumption_4} hold true and the client's learning rate is set as $\alpha=\frac{1}{\lambda}$, our proposed framework, denoted as \alg, uses the Median aggregation rule to merge synthetic updates and updates contributed by clients. Assuming $\upsilon > 0$, if $\beta=\frac{f}{n+m}$ satisfies $\beta + \epsilon \le \frac{1}{2} - \upsilon$, then after $T$ rounds of global training, the probability of achieving the following outcome is guaranteed to be at least $1 - \frac{4d}{(1 + (n+m) \lambda \varpi Q)^d}$:
\begin{align}
 \left\| \bm{\theta}^T - \bm{\theta}^* \right\| \le 
  (1-\frac{\mu}{\mu+\lambda})^T   \left\| \bm{\theta}^0 - \bm{\theta}^* \right\| + \frac{2 B_2}{\mu}, \nonumber
  \end{align}

where $\epsilon=\frac{0.4748  \zeta}{\sqrt{Q}} + \sqrt{\frac{d\log(1+(n+m) \lambda \varpi Q)}{(n+m)(1-\beta)}} $,
$B_2=\frac{2\sqrt{2}}{(n+m) Q} + \frac{2 \sqrt{\pi} \sigma (\beta + \epsilon) \exp(\frac{1}{2}(\Phi^{-1}(1-\upsilon))^2)}{\sqrt{Q}}$, and $\Phi$ represents the cumulative distribution function of the standard Gaussian distribution.

\end{thm}

\begin{proof}
   The proof is relegated to Appendix~\ref{sec:appendix_2}.
\end{proof}

\begin{remark} 
In our theoretical analysis, we adopt simplifying assumptions commonly used in the FL community~\cite{yin2018byzantine,ChenPOMACS17,chu2022securing,karimireddy2020byzantine}.
However, we acknowledge that these assumptions may not fully capture real-world conditions. To assess \alg's sensitivity, we conduct additional experiments by relaxing certain assumptions, such as Assumption~\ref{assumption_1}, which pertains to the non-convexity of deep neural networks. Specifically, we test \alg on two CNN architectures and the more complex ResNet-18 model~\cite{he2016deep}, neither of which satisfies Assumption~\ref{assumption_1}. Extensive experimental results show that our proposed \alg remains secure even when certain assumptions are partially relaxed.
\end{remark}

%!TEX root = main.tex

\section{Evaluation} \label{sec:exp}

\subsection{Experimental Setup}

\subsubsection{Datasets} In our experiments, we use six distinct datasets from various domains, which encompass MNIST~\cite{lecun2010mnist}, Fashion-MNIST~\cite{xiao2017online}, Human Activity Recognition (HAR)~\cite{anguita2013public}, Purchase~\cite{PurchaseDataset}, Large-scale CelebFaces Attributes (CelebA)~\cite{liu2015faceattributes}, and CIFAR-10~\cite{krizhevsky2009learning}.

\myparatight{a) MNIST~\cite{lecun2010mnist}}The MNIST dataset consists of 60,000 training images and 10,000 testing images, encompassing a total of 10 unique classes.

\myparatight{b) Fashion-MNIST~\cite{xiao2017online}}The Fashion-MNIST dataset comprises a total of 70,000 fashion images. The training dataset contains 60,000 images, and the testing set contains 10,000 images. Each image in Fashion-MNIST is assigned to one of 10 classes.

\myparatight{c) Human Activity Recognition (HAR)~\cite{anguita2013public}}The HAR dataset is a practical dataset used for predicting human activities. It includes data collected from 30 users who used smartphones in their daily routines, totaling 10,299 instances. Each instance consists of 561 features and is classified into one of six distinct categories. Following the approach in~\cite{cao2020fltrust}, in our experiments, we randomly assign 75\% of the data from each user for training, while the remaining 25\% is kept aside for testing.

\myparatight{d) Purchase~\cite{PurchaseDataset}}The purchase classification dataset is imbalanced, containing 197,324 examples, each characterized by 600 binary attributes distributed among 100 different categories. In our experiments, we randomly choose a subset of 150,000 examples for training our models, reserving the remaining 47,324 examples for testing purposes.

% \myparatight{e) CIFAR-10~\cite{krizhevsky2009learning}}
% The CIFAR-10 dataset consists of 10 classes. It includes 50,000 color images for training and 10,000 color images for testing, with each example belonging to one of the ten classes.

\myparatight{e) Large-scale CelebFaces Attributes (CelebA)~\cite{liu2015faceattributes}}This dataset involves a binary classification task to determine whether the person in an image is smiling or not. The CelebA dataset includes 177,480 training examples and 22,808 testing examples.

% The CelebA dataset, used for binary classification to identify if a person in an image is smiling, contains 177,480 training examples and 22,808 testing examples.

% This dataset involves a binary classification task to determine whether the person in an image is smiling or not. The CelebA dataset includes 177,480 training examples and 22,808 testing examples.

\myparatight{f) CIFAR-10~\cite{krizhevsky2009learning}}CIFAR-10 is a color image classification dataset with 10 distinct classes. It includes a total of 60,000 images, with 50,000 used for training and the remaining 10,000 for testing.

% This is a color image classification dataset with 10 distinct classes. 
% The dataset comprises 60,000 images, split into 50,000 for training and 10,000 for testing.
% It includes a total of 60,000 images, with 50,000 used for training and the remaining 10,000 for testing.

\subsubsection{Poisoning Attacks} 
% In our experiments, we examine various poisoning attacks, comprising six untargeted attacks 
%
By default, our experiments examine six untargeted attacks
(label flipping attack~\cite{tolpegin2020data}, Gaussian attack~\cite{blanchard2017machine}, Trim attack~\cite{fang2020local}, Krum attack~\cite{fang2020local}, Min-Max attack~\cite{shejwalkar2021manipulating}, Min-Sum attack~\cite{shejwalkar2021manipulating}) and one targeted attack (Scaling attack~\cite{bagdasaryan2020backdoor,cao2020fltrust}).
Note that we also consider 
two additional targeted attacks
and three more sophisticated attacks in Section~\ref{sec:discussion_limitation}.
By evaluating our method with a total of 12 representative poisoning attacks, covering both untargeted and targeted strategies, we ensure a comprehensive evaluation that reflects real-world scenarios and challenges, rigorously testing our method’s robustness across diverse attack models. These attacks are selected because they are widely studied in the literature~\cite{fang2020local,shejwalkar2021manipulating,tolpegin2020data,cao2020fltrust,nguyen2022flame,shejwalkar2022back} and represent a range of commonly observed methods posing substantial threats to FL systems.

\myparatight{a) Label flipping (LF) attack~\cite{tolpegin2020data}}In the LF attack, the attacker modifies the labels of training data on malicious clients. Specifically, for a training example originally labeled as $y$, the attacker changes it to $M - y - 1$, where $M$ represents the total number of labels.

\myparatight{b) Gaussian attack ~\cite{blanchard2017machine}}In this specific attack, each malicious client sends a randomly generated vector to the server. These vectors are sampled from a Gaussian distribution with a mean of 0 and a variance of 200.

\myparatight{c) Trim attack~\cite{fang2020local}}The Trim attack is a strategy targeting aggregation Trimmed-mean and Median aggregation rules. 
In this attack, the attacker carefully designs the model updates on malicious clients to ensure that the post-attack model update diverges significantly from its pre-attack one.

\myparatight{d) Krum attack~\cite{fang2020local}}The Krum attack is an advanced strategy aimed at exploiting the Krum aggregation rule. The attacker strategically crafts the model updates on malicious clients to influence the Krum rule into choosing the malicious update as the final aggregated outcome.

\myparatight{e) Min-Max attack~\cite{shejwalkar2021manipulating}}In the Min-Max attack, the attacker tailors local model updates on malicious clients to achieve their objectives, ensuring that these updates closely resemble benign updates. Specifically, the maximum distance between a malicious local model update and any benign local model update is smaller than the maximum distance between any two benign local model updates.

\myparatight{f) Min-Sum attack~\cite{shejwalkar2021manipulating}}The Min-Sum attack is another attack model that is agnostic to aggregation rules. Unlike the Min-Max attack, in the Min-Sum attack, the attacker ensures that the sum of distances between a malicious local model update and all benign local model updates does not exceed the maximum sum of distances between any two benign updates.

\myparatight{g) Scaling attack~\cite{bagdasaryan2020backdoor,cao2020fltrust}}The Scaling attack is a type of targeted attack where the attacker initially augments the local training data of malicious clients by introducing backdoor triggers into duplicated data copies. Subsequently, these malicious clients train their local models using the augmented training data and further amplify their local model updates before transmitting them to the server.

\subsubsection{Compared Aggregation Rules} 
By default, we compare our proposed \alg with the following seven aggregation rules.

\myparatight{a) FedAvg~\cite{mcmahan2017communication}}FedAvg is a non-robust aggregation method where the server aggregates received local model updates by computing the average of all updates.

% \myparatight{1) Krum~\cite{blanchard2017machine}}
% The Krum aggregation rule involves the server selecting and outputting a local model update which has the smallest total distance to a a subset of its neighboring updates.

\myparatight{b) Trimmed-mean (Trim-mean)~\cite{yin2018byzantine}}Trimmed-mean is an aggregation method applied per coordinate. For each dimension, the server removes the largest $c$ and smallest $c$ elements, then computes the average of the remaining values. Here, $c$ is referred to as the trim parameter.

\myparatight{c) GAS + Trim-mean~\cite{liu2023byzantine}}Upon receiving the local model updates from all clients, the server splits each update into multiple parts. The server then applies the Trimmed-mean~\cite{yin2018byzantine} aggregation rule to compute the aggregated result for each part. Subsequently, the server calculates an identification score for each client and selects the $n - f$ local model updates with the lowest identification scores for aggregation by taking the average of these $n - f$ updates, where $f$ represents the total number of malicious clients.

\myparatight{d) Gaussian + Trim-mean}Upon receiving \( n \) local model updates from clients, the server creates $m$ synthetic updates. Each \( k \)-th dimension of these updates follows a Gaussian distribution \( \mathcal{N}(\mu,\, \sigma^2) \), where \( \mu \) and \( \sigma \) are the mean and standard deviation of \( \{\bm{g}_i^t[k] : i \in [n]\} \) with \( \bm{g}_i^t[k] \) representing the \( k \)-th dimension of \( \bm{g}_i^t \), and $k \in [d]$. Subsequently, the server uses the Trimmed-mean~\cite{yin2018byzantine} aggregation rule to merge these $m$ synthetic updates with the \( n \) received updates.

\myparatight{e) Median~\cite{yin2018byzantine}}The Median is another aggregation rule that is applied on a per-coordinate basis. For each dimension, the server calculates the median value from all the received model updates.

% The Median aggregation rule computes the median of received model updates for each coordinate.

% The Median is another aggregation rule that is applied on a per-coordinate basis. For each dimension, the server calculates the median value from all the received model updates.

\myparatight{f) GAS + Median~\cite{liu2023byzantine}}Similar to the GAS + Trim-mean aggregation rule, the server in the GAS + Median rule also divides each local model update into multiple parts. However, in this method, the server applies the Median~\cite{yin2018byzantine} aggregation rule to each part. Then, the server calculates the final aggregated update by averaging the $n - f$ local model updates with the lowest identification scores.

\myparatight{g) Gaussian + Median}For this method, the server follows the same procedure as in the Gaussian + Trim-mean method to produce the $m$ synthetic updates. The only distinction lies in the aggregation technique; here, the server employs the Median~\cite{yin2018byzantine} aggregation rule to aggregate the $m$ synthetic updates with the \( n \) local model updates from clients.

\subsubsection{Evaluation Metrics} Two evaluation metrics are explored in this paper.

\myparatight{a) Testing error rate}The testing error rate reflects the percentage of test instances incorrectly classified by the global model. A lower testing error rate signifies a stronger defense.

\myparatight{b) Attack success rate}The attack success rate is determined by the proportion of targeted examples predicted as the labels chosen by the attacker. A lower attack success rate indicates more effective defense.

\subsubsection{Non-IID Setting}
In FL, a distinctive aspect is that clients' local training data are not independently and identically distributed (Non-IID). In our study, we adopt the following method to simulate the Non-IID setting as described in~\cite{fang2020local}. For a dataset with \( M \) classes, clients are initially randomly grouped into \( M \) clusters. A training example labeled \( y \) is then assigned to clients in cluster \( y \) with probability \( h \), and to other clusters with probability \( \frac{1-h}{M-1} \). A higher value of \( h \) indicates greater Non-IID characteristics in the clients' training data. 
For the MNIST and Fashion-MNIST datasets, we set \( h = 0.5 \).
It is important to note that we do not simulate the Non-IID setting for the HAR, Purchase, and CelebA datasets, as these datasets inherently exhibit heterogeneity.

\begin{table*}[htbp]
 % \vspace{-1.06in}
  \centering
  \addtolength{\tabcolsep}{-1.985pt}
  \caption{Results of different FL methods on MNIST, Fashion-MNIST, HAR, and Purchase datasets. The results of Scaling attack are shown as ``testing error rate / attack success rate''.}
        \vspace{-0.06in}
   \subfloat[MNIST dataset.]
   {
    \begin{tabular}{l|cccccccc}
    \hline
    Aggregation rule & No attack   & LF attack    & Gaussian attack & Trim attack & Krum attack & Min-Max attack & Min-Sum attack & \multicolumn{1}{c}{Scaling attack} \\
     \hline
    FedAvg  &   0.05    &  0.07     &   0.90    &   0.32    &  0.10     &  0.90     &   0.90    &  0.64 / 0.70 \\
    \hline
    Trim-mean &  0.06     &  0.06     &  0.06     & 0.27      &   0.08    &  0.19     &   0.13    & 0.13 / 0.02 \\
    GAS + Trim-mean &  0.05  &  0.05   &  0.11   &   0.29   & 0.07    & 0.10   &  0.11    &  0.43 / 0.47 \\
    Gaussian + Trim-mean & 0.05      &  0.11     &  0.91     &    0.91   &   0.05    &   0.08    &  0.06     & 0.91 / 1.00 \\
    \rowcolor{greyL}
    \alg + Trim-mean &   0.05    &  0.05     & 0.05      &   0.05    &  0.05     &   0.05    &   0.05    & 0.05 / 0.02 \\
     \hline
    Median &   0.05    &   0.09    & 0.16      &   0.23    &   0.17    &   0.19    &   0.23    & 0.05 / 0.02 \\
    GAS + Median &   0.05    &   0.05    &  0.12     &  0.26     &   0.06    &  0.10     &  0.10     & 0.59 / 0.65 \\
    Gaussian + Median &   0.05    &   0.90    &    0.90     & 0.90    &  0.05     &  0.14     &   0.14    &  0.91 / 1.00 \\
    \rowcolor{greyL}
    \alg + Median &  0.05    &    0.05     &   0.05  &  0.05    &  0.07     &   0.05     &   0.05     & 0.06 / 0.02 \\
     \hline
    \end{tabular}%
    }
   \\
        \vspace{-0.06in}
       \subfloat[Fashion-MNIST dataset.]
   {
    \begin{tabular}{l|cccccccc}
    \hline
    Aggregation rule & No attack   & LF attack    & Gaussian attack & Trim attack & Krum attack & Min-Max attack & Min-Sum attack & \multicolumn{1}{c}{Scaling attack} \\
     \hline
    FedAvg  &   0.18    &   0.21    &  0.21     &  0.38     &  0.21     &   0.30    & 0.32      &  0.63 / 0.68 \\
    \hline
    Trim-mean &   0.21   & 0.29      & 0.26  &  0.47     & 0.32      &   0.24   &   0.23    & 0.26 / 0.01 \\
    GAS + Trim-mean & 0.21   &  0.21   & 0.24    &  0.50      &    0.26   & 0.26   &  0.23   & 0.67 / 0.62 \\
    Gaussian + Trim-mean &  0.18   &   0.90    &   0.90    &   0.90      &   0.20    &   0.32    &    0.90   & 0.90 / 1.00 \\
       \rowcolor{greyL}
    \alg + Trim-mean &  0.18     & 0.18      &  0.18     &  0.19     &  0.20     &  0.19      &  0.18     & 0.22 / 0.03 \\
     \hline
    Median &   0.23   &  0.27     &  0.27     &   0.35    &  0.29     &  0.31     &  0.29     & 0.29 / 0.03 \\
    GAS + Median &   0.20    &  0.20     &  0.24     &  0.44     & 0.25      &  0.26     &   0.23    & 0.60 / 0.64 \\
    Gaussian + Median &  0.23     &  0.90     &    0.35      &   0.90       &  0.90        &   0.90     &   0.90     & 0.90 / 1.00 \\
     \rowcolor{greyL}
    \alg + Median &   0.18   &  0.20     &  0.18   &  0.18      &  0.21     &  0.19     & 0.20      &  0.23 / 0.03 \\
     \hline
    \end{tabular}%
    }
    \\
        \vspace{-0.06in}
       \subfloat[HAR dataset.]
   {
    \begin{tabular}{l|cccccccc}
    \hline
    Aggregation rule & No attack   & LF attack    & Gaussian attack & Trim attack & Krum attack & Min-Max attack & Min-Sum attack & \multicolumn{1}{c}{Scaling attack} \\
     \hline
    FedAvg  &   0.05    & 0.12  &  0.13 & 0.38  & 0.06  &  0.14   & 0.17      & 0.05 / 0.91 \\
    \hline
    Trim-mean &   0.07    & 0.07  & 0.07  &  0.26  & 0.09   & 0.16   &  0.16     &  0.07 / 0.02  \\
    GAS + Trim-mean & 0.06   &  0.09 & 0.12    &  0.35   &  0.06    &  0.12     &   0.15  & 0.07 / 0.92 \\
    Gaussian + Trim-mean &  0.05     & 0.12   &  0.41   &0.24 & 0.06   &  0.17 &  0.17  & 0.62 / 0.01 \\
    \rowcolor{greyL}
    \alg + Trim-mean &  0.05   &  0.08 & 0.05   & 0.08   &  0.06 &   0.09    & 0.09 & 0.05 / 0.01 \\
     \hline
    Median & 0.07   & 0.08  &  0.09  &  0.16  & 0.08  &  0.17 &   0.15 & 0.07 / 0.02 \\
    GAS + Median &  0.07  & 0.10   & 0.11   & 0.41  & 0.06 &  0.16 &    0.17   & 0.07 / 0.88 \\
    Gaussian + Median & 0.06   & 0.14   &   0.11  & 0.30  & 0.06  & 0.15      & 0.15   & 0.08 / 0.10 \\
    \rowcolor{greyL}
    \alg + Median &  0.05   &  0.09   & 0.06  & 0.09 &  0.06  & 0.09  &  0.09   & 0.05 / 0.02 \\
     \hline
    \end{tabular}%
    }
    \\
        \vspace{-0.06in}
       \subfloat[Purchase Dataset.]
   {
    \begin{tabular}{l|cccccccc}
    \hline
    Aggregation rule & No attack   & LF attack    & Gaussian attack & Trim attack & Krum attack & Min-Max attack & Min-Sum attack & \multicolumn{1}{c}{Scaling attack} \\
     \hline
    FedAvg  &  0.17     &  0.32     &  0.72     &  0.60     &  0.20     &  0.34     &  0.34     &  0.99 / 0.41 \\
    \hline
    Trim-mean & 0.21      &  0.28     &  0.21      &  0.47    &  0.26     &   0.40    &   0.40    &  0.25 / 0.02 \\
    GAS + Trim-mean &    0.18    &   0.30    &   0.79     &   0.70    &  0.21     &  0.37     & 0.35      & 0.23 / 0.01 \\
    Gaussian + Trim-mean & 0.17      & 0.20  &  0.96 &  0.58     & 0.17   &    0.38   &   0.38    &  0.99 / 1.00 \\
    \rowcolor{greyL}
    \alg + Trim-mean & 0.17   & 0.20     & 0.17   &  0.21    & 0.19      & 0.21      &  0.22     & 0.18 / 0.02 \\
     \hline
    Median &  0.23     &  0.32 &  0.25 & 0.48 &  0.30   & 0.49  &  0.49  & 0.26 / 0.03 \\
    GAS + Median &  0.17  &   0.30  &  0.70 &  0.72 &  0.21 &  0.37 &   0.37  &  0.20 / 0.02\\
    Gaussian + Median &0.18 & 0.22 & 0.58 & 0.58  &  0.18  & 0.29  &  0.29  &  0.99 / 1.00 \\
    \rowcolor{greyL}
    \alg + Median &   0.18    &  0.23   &  0.19  &   0.24    &   0.24   & 0.24     &  0.24  & 0.21 / 0.03 \\
     \hline
    \end{tabular}%
    }
     \label{result_all_datasets}%
     % \vspace{-.15in}
\end{table*}%

\subsubsection{Parameter Settings} 
We consider 100 clients each for the MNIST, Fashion-MNIST, and 
CIFAR-10 datasets ($n=100$), 40 clients for the Purchase dataset, 20 clients for the CelebA dataset, and 30 clients in total for the HAR dataset, where each real-world user is treated as a client. By default, we assume 20\% of the clients are malicious. 
For the MNIST, Fashion-MNIST, and CelebA datasets, we train a convolutional neural network (CNN) whose architecture is detailed in Table~\ref{cnn_arch} in Appendix. 
The HAR dataset is trained using a logistic regression classifier. The Purchase dataset employs a fully connected neural network as the global model architecture with one hidden layer consisting of 1,024 neurons and Tanh activation function.
We train a ResNet-18~\cite{he2016deep} model for the CIFAR-10 dataset.
We conduct training for 2,000 rounds on the MNIST dataset, 3,000 rounds on Fashion-MNIST, 1,000 rounds each on the HAR and Purchase datasets, 1,000 rounds on CelebA, 
and 1,000 rounds on CIFAR-10.
The batch sizes are set to 32 for MNIST and Fashion-MNIST, 32 for HAR, 128 for Purchase, 20 for CelebA, 
and 40 for CIFAR-10.
The corresponding learning rates are 1/3,200 for MNIST, Fashion-MNIST, and HAR, 1/1,280 for Purchase, 1/20,000 for CelebA,
and 0.005 for CIFAR-10 dataset.
Following~\cite{blanchard2017machine,yin2018byzantine,liu2023byzantine}, we set $c=f$ by default.
In the GAS approach, we use the parameter as recommended in~\cite{liu2023byzantine}.
Unless stated otherwise, we assume that all clients participate in the training process in every round.
In the \alg framework we propose, the server generates $\frac{n}{2}$ synthetic updates in each training round by default for MNIST, Fashion-MNIST, HAR, Purchase, and CelebA.
The results are primarily reported using the MNIST dataset by default.

\subsection{Experimental Results}

\myparatight{Our \alg is effective}Table~\ref{result_all_datasets} displays the performance of various FL methods under different poisoning attacks on the MNIST, Fashion-MNIST, HAR, and Purchase datasets. The results for the CelebA and CIFAR-10 datasets are presented in Table~\ref{result_all_datasets_CelebA_cifar10} in Appendix.
The term ``No attack'' indicates that all client in the FL system are benign without any malicious clients. The results of the Scaling attack are shown in the form of ``testing error rate / attack success rate''. The terms ``\alg + Trim-mean'' and ``\alg + Median'' describe our method where the server combines clients' local model updates and synthetic updates using Trimmed-mean and Median aggregation rules, respectively. Notably, under benign conditions, our \alg framework mirrors the performance of FedAvg. For instance, both ``\alg + Trim-mean'' and ``\alg + Median'' achieve similar testing error rates as FedAvg on the MNIST, Fashion-MNIST, and HAR datasets when no malicious clients are present. However, in the presence of malicious clients attempting to compromise the system, our \alg framework is uniquely capable of defending against such attacks. For example, on the MNIST dataset, our approach under attack performs comparably to FedAvg in a non-attack scenario. Nonetheless, existing Byzantine-robust foundational aggregation rules like Trim-mean and Median show inherent weaknesses to poisoning attacks; for example, the testing error rate for Trim-mean on the Fashion-MNIST dataset escalates from 0.21 under no attack to 0.47 under the Trim attack. Similarly, more complex aggregation schemes such as ``GAS + Trim-mean'' and ``GAS + Median'' are also susceptible. 
Even on a complex dataset and with a complicated architecture like the ResNet-18 model for the CIFAR-10 image classification task, our proposed \alg can protect FL against poisoning attacks. In contrast, other aggregation rules, such as Trim-mean and Median, show higher testing error rates of 0.79 and 0.84, respectively, under the Trim attack (see Table~\ref{result_all_datasets_cifar10} in Appendix).
Our framework not only combats untargeted attacks but also protects effectively against targeted threats like the Scaling attack, as evidenced in Table~\ref{result_all_datasets} and Table~\ref{result_all_datasets_CelebA_cifar10}.
Note that in our proposed \alg, the server produces synthetic updates after collecting local model updates from clients. Using established Byzantine-robust aggregation rules, the server then combines these augmented model updates to reduce update variance. This is further supported by our experimental findings. For instance, under the Trim attack on the MNIST dataset, the mean variances for the Median and ``\alg + Median'' approaches are 3.39 and 0.41, respectively. This indicates that \alg effectively reduces the variance in updates.

\begin{table*}[!t]
  \centering
  \addtolength{\tabcolsep}{-1.985pt}
  \caption{Results when the server uses either the synthetic update alone or FedAvg to merge the augmented model updates.}
    \begin{tabular}{l|cccccccc}
    \hline
    Aggregation rule & No attack   & LF attack    & Gaussian attack & Trim attack & Krum attack & Min-Max attack & Min-Sum attack & \multicolumn{1}{c}{Scaling attack} \\
     \hline
    Synthetic only  & 0.08  &  0.12  &  0.08  &  0.12 & 0.91  &   0.12  & 0.08  & 0.91 / 1.00 \\
    \alg + FedAvg   & 0.05 & 0.14  & 0.87  &  0.12 &0.06  &  0.07 & 0.08 & 0.46 / 0.59 \\
     \hline
    \end{tabular}%
     \label{result_table_synthetic_only}
     % \vspace{-.15in}
\end{table*}%

\begin{table*}[!t]
  \centering
    \addtolength{\tabcolsep}{-0.285pt}
  \caption{Results of complex Byzantine-robust aggregation rules.}
          \vspace{-0.06in}
    \begin{tabular}{l|cccccccc}
    \hline
    Aggregation rule & No attack   & LF attack    & Gaussian attack & Trim attack & Krum attack & Min-Max attack & Min-Sum attack & \multicolumn{1}{c}{Scaling attack} \\
     \hline
     Krum  & 0.10 & 0.10  &  0.10 & 0.10  & 0.90 & 0.11  & 0.11 & 0.10 / 0.01\\
     FoolsGold  & 0.09 & 0.12  & 0.09  & 0.37 & 0.12& 0.25  & 0.22 &  0.13 / 0.05 \\
     FLAME  & 0.07 & 0.08  & 0.08  & 0.17 & 0.08& 0.07  & 0.07 & 0.08 / 0.02 \\
     \hline
    \end{tabular}%
     \label{result_complex}
    % \vspace{-.05in}
\end{table*}%

It is important to highlight that in our proposed framework, the server employs robust, coordinate-wise aggregation methods like Trimmed-mean or Median, aiming to minimize outlier effects in each dimension. That is to say, the server does not directly use the selected local model update (one synthetic update) to update the global model. 
``Synthetic only" in Table~\ref{result_table_synthetic_only} illustrates the performance when the global model is updated exclusively with the selected local
model update. 
This approach, as shown in Table~\ref{result_table_synthetic_only}, is particularly susceptible to Krum and Scaling attacks, with testing error rates soaring from 0.08 under normal conditions to 0.91 during these attacks. 
% The increase in error rates occurs because the chosen update might still include outliers in some dimensions.
%
%
``\alg + FedAvg'' in Table~\ref{result_table_synthetic_only} refers to the method in which the server applies the FedAvg rule to merge synthetic updates with clients' model updates. We observe that this method is susceptible to Gaussian attack, as the synthetic updates can include extreme values in some dimensions. Therefore, we require robust aggregation rules to filter out these extreme values.

In recent years, a variety of new and sophisticated robust aggregation rules have been introduced. Table~\ref{result_complex} displays the testing error rate and attack success rate for three complex and representative Byzantine-robust aggregation rules on the MNIST dataset: Krum~\cite{blanchard2017machine}, FoolsGold~\cite{fung2018mitigating}, and FLAME~\cite{nguyen2022flame}. As observed in Table~\ref{result_complex}, these advanced robust aggregation rules remain susceptible to poisoning attacks. For example, the Krum aggregation rule is fundamentally vulnerable to the Krum attack, and the FLAME method is susceptible to the Trim attack.
Comparing Table~\ref{result_all_datasets} with Table~\ref{result_complex}, it is evident that our proposed \alg framework demonstrates greater robustness compared to these sophisticated aggregation rules.

\begin{figure*}[!t]
	\centering
	\includegraphics[scale = 0.5]{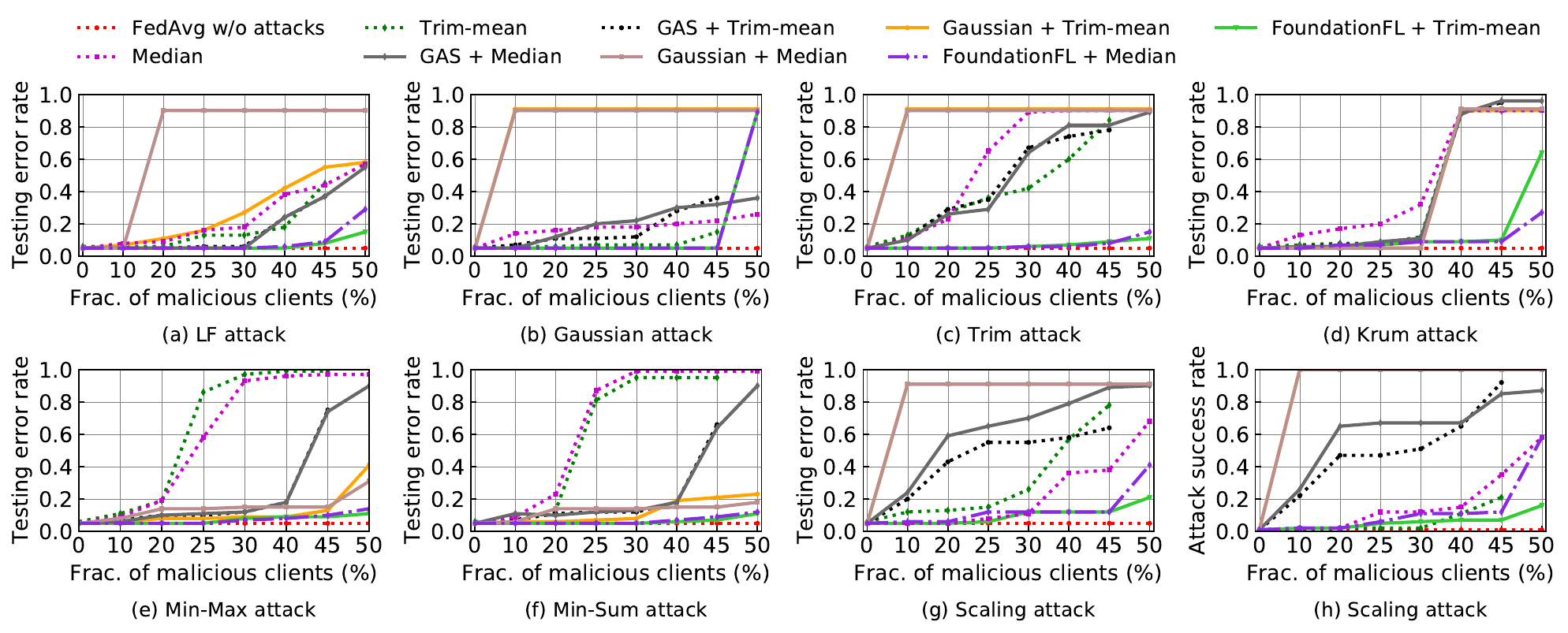}
	\caption{Impact of fraction of malicious clients.}
	\label{fig_attack_size}
		\vspace{-.1in}
\end{figure*}

\myparatight{Impact of fraction of malicious clients}Fig.~\ref{fig_attack_size} illustrates the impact of poisoning attacks on the performance of various FL aggregation methods within the MNIST dataset, where the proportion of malicious clients increases from 0\% to 
50\%, 
with a total client base of 100. 
The fraction of malicious clients is computed as \( f/n \), with \( f \) representing the number of malicious clients and \( n \) the total client count.
Note that for the Trim-mean and ``GAS + Trim-mean'' aggregation rules, the fraction of malicious clients ranges only from 0 to 45\%, as these two methods require the number of malicious clients to be less than half of the total clients.
From Fig.~\ref{fig_attack_size}, it's evident that our proposed methods,  ``\alg + Trim-mean'' and ``\alg + Median'', remain robust against poisoning attacks, even when up to 45\% of the clients are malicious. For example, the testing error rate for ``\alg + Trim-mean'' only slightly increases from 0.05 with no attack to 0.06 under the strong Trim attack when 30\% of clients are malicious. 
With 45\% malicious clients, ``\alg + Trim-mean'' and ``\alg + Median'' maintain error rates no higher than 0.12 across different attacks.
Conversely, with just 10\% malicious clients, the testing error rate for ``Gaussian + Median'' reaches 0.90 under the Trim attack.

\begin{figure*}[!t]
	\centering
	\includegraphics[scale = 0.5]{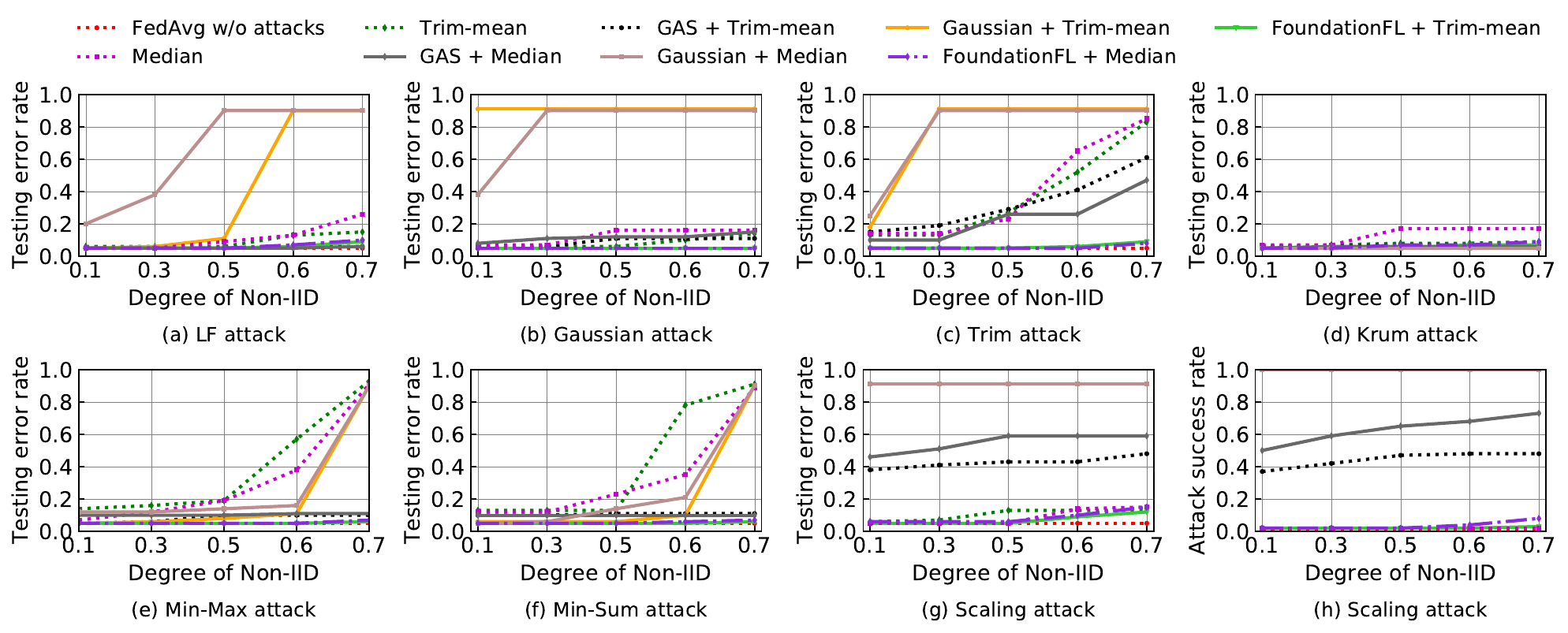}
	\caption{Impact of degree of Non-IID.}
	\label{fig_non_iid}
		\vspace{-.1in}
\end{figure*}

\begin{figure*}[!t]
	\centering
	\includegraphics[scale = 0.5]{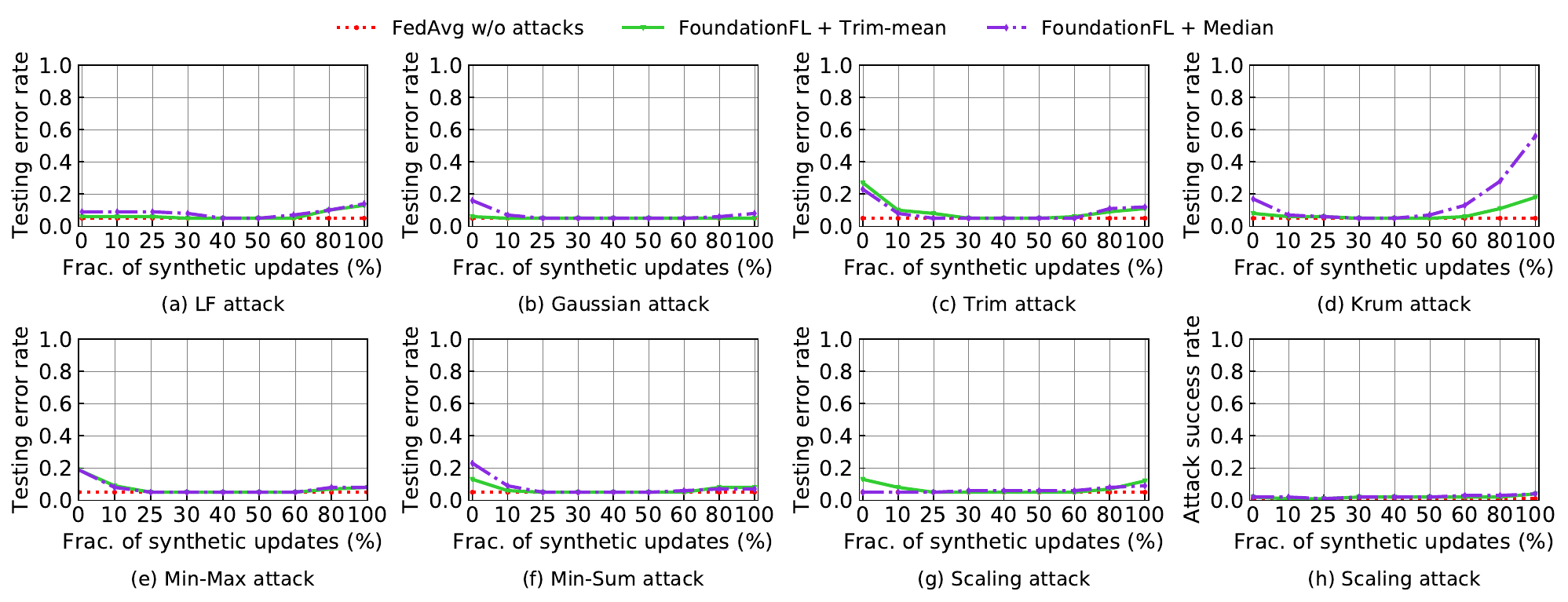}
	\caption{Impact of fraction of synthetic updates.}
	\label{fig_fake_number}
		\vspace{-.1in}
\end{figure*}

\begin{figure*}[!t]
	\centering
	\includegraphics[scale = 0.5]{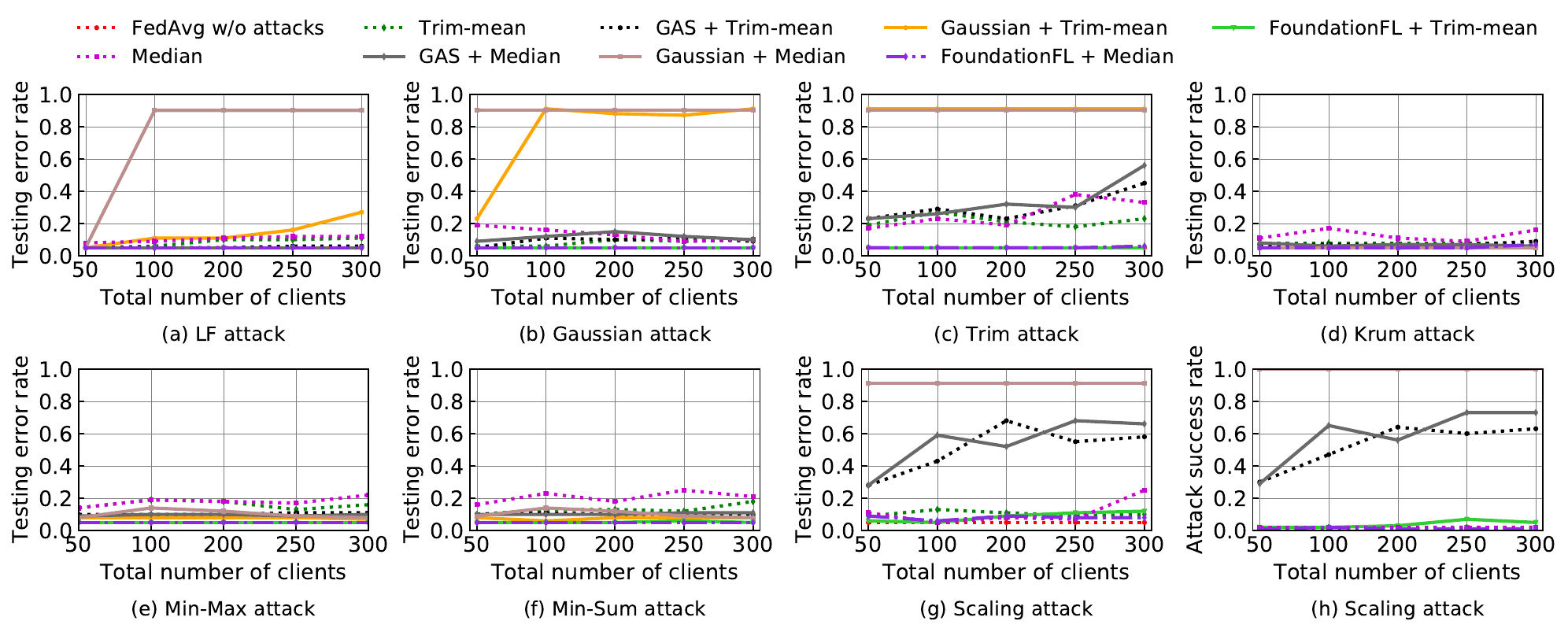}
	\caption{Impact of total number of clients.}
	\label{fig_total_client}
		\vspace{-.1in}
\end{figure*}

\myparatight{Impact of degree of Non-IID}In FL, a distinct characteristic is the Non-IID nature of clients' local training data. When this data is notably diverse, the attacker find it easier to craft malicious model updates that appear benign yet can significantly disrupt the targeted FL system. This is particularly problematic when the system employs Byzantine-robust foundational aggregation rules like Trimmed-mean or Median to merge these updates. Our findings, depicted in Fig.~\ref{fig_non_iid}, investigates the influence of Non-IID data heterogeneity on the efficacy of various FL methods. The degree of Non-IIDness explored ranges from 0.1 to 0.7, with other parameters set to default values. The findings illustrated in Fig.~\ref{fig_non_iid} demonstrate that despite the high heterogeneity in clients' training data, our proposed \alg framework effectively shields against diverse poisoning attacks.

\myparatight{Impact of fraction of synthetic updates}In our framework, once the server collects $n$ local model updates from clients during each global training round, it proceeds to generate $m$ synthetic updates. These $n+m$ updates are then aggregated using either the Trim-mean or Median method. This section explores how the proportion of synthetic updates, calculated as $m/n$, influences our proposed \alg. The results are displayed in Fig.~\ref{fig_fake_number}. From these findings, it is evident that our methods, ``\alg + Trim-mean'' and ``\alg + Median'', exhibit robustness against variations in the proportion of synthetic updates. 
% Even when the fraction of synthetic updates varies between 25\% and 60\%, our proposed methods maintain their effectiveness.
%
%
It is important to note that when the fraction of synthetic updates is 0\%, our ``\alg + Trim-mean'' and ``\alg + Median'' methods are equivalent to the ``Trim-mean'' and ``Median'' rules, respectively. From Fig.~\ref{fig_fake_number}, we observe that \alg requires between 10\% and 60\% synthetic updates to effectively defend against various poisoning attacks. This is because too few synthetic updates fail to mitigate the influence of malicious clients, while too many could overshadow the benign updates within the system.

\begin{table*}[htbp]
  \centering
  \addtolength{\tabcolsep}{-1.985pt}
  \caption{Results of different defenses on a production FL system.}
          \vspace{-0.06in}
    \begin{tabular}{l|cccccccc}
    \hline
    Aggregation rule & No attack   & LF attack    & Gaussian attack & Trim attack & Krum attack & Min-Max attack & Min-Sum attack & \multicolumn{1}{c}{Scaling attack} \\
     \hline
    FedAvg  & 0.05 & 0.08  & 0.06 & 0.30 & 0.08 & 0.75 & 0.75   & 0.71 / 0.69  \\
    \hline
    Trim-mean  & 0.06  & 0.09   & 0.06 & 0.27 & 0.09 & 0.28 &  0.13  & 0.13 / 0.01  \\
    GAS + Trim-mean  & 0.06 & 0.07  &  0.07 & 0.32 & 0.07 & 0.10 &  0.10  & 0.32 / 0.28  \\
    Gaussian + Trim-mean   & 0.05 & 0.19  & 0.91 & 0.91 & 0.06 & 0.06 & 0.08   & 0.89 / 1.00  \\
    \rowcolor{greyL}
    \alg + Trim-mean  & 0.05  & 0.05  & 0.05 & 0.05 & 0.05 & 0.05 &  0.05  &  0.05 / 0.01 \\
     \hline
    Median  & 0.06 & 0.08 &0.06  & 0.38 & 0.11 & 0.21 & 0.16   &  0.06 / 0.02 \\
    GAS + Median  & 0.05 &  0.05 & 0.07 & 0.47 &  0.17 & 0.08 &  0.08  & 0.55 / 0.53  \\
    Gaussian + Median  & 0.05  & 0.91  &   0.91&  0.91 & 0.05 & 0.08 &  0.07  & 0.91 / 1.00  \\
    \rowcolor{greyL}
    \alg + Median  & 0.05 & 0.06  & 0.05 & 0.05  & 0.07 & 0.05 &  0.05  & 0.05 / 0.02  \\
     \hline
    \end{tabular}%
     \label{result_table_production_fl}%
\end{table*}%

\begin{table*}[htbp]
  \centering
  \caption{Results of ``\alg + Trim-mean'' in scenarios where the server lacks knowledge of the number of malicious clients $f$, the trim parameter is set to $c=30$ for $f=10$ and $f=20$.}
          % \vspace{-0.06in}
    \begin{tabular}{l|cccccccc}
    \hline
    $f$ & No attack   & LF attack    & Gaussian attack & Trim attack & Krum attack & Min-Max attack & Min-Sum attack & \multicolumn{1}{c}{Scaling attack} \\
     \hline
     $f=10$  & 0.05  & 0.07  &  0.06  & 0.07 &0.06 & 0.06  & 0.06  & 0.05 / 0.02 \\
     $f=20$  &  0.05 & 0.06 & 0.05  & 0.07 & 0.06 & 0.05   &  0.05 & 0.05 / 0.02 \\
     Estimate  & 0.05 & 0.08  & 0.06  & 0.08 & 0.06 &  0.06  & 0.05  & 0.06 / 0.02 \\
     \hline
    \end{tabular}%
     \label{result_malicious_number_unknow}
\end{table*}%

\myparatight{Impact of total number of clients}Fig.~\ref{fig_total_client} displays results across a varying total number of clients from 50 to 300, with a consistent fraction of 20\% malicious clients and other parameters at default settings, using the MNIST dataset. The figure demonstrates that our methods, ``\alg + Trim-mean'' and ``\alg + Median'', consistently defend against poisoning attacks effectively across all client scales. Conversely, traditional robust aggregation methods like ``GAS + Median'' fail to adequately counteract the effects of these attacks. Notably, as the total number of clients ranges from 50 to 300, the testing error rates for ``GAS + Median'' remain significantly high under a Trim attack.

\myparatight{Results of various defense methods on an alternative CNN architecture}In this part, we demonstrate the robustness of various aggregation methods on an alternative CNN architecture, with details of this architecture provided in Table~\ref{cnn_arch_diff} in Appendix. Results for the different defense methods are presented in Table~\ref{result_table_diff_arch} in Appendix. From Table~\ref{result_table_diff_arch}, we observe that our proposed \alg can effectively defend against various poisoning attacks, even with this CNN architecture. For instance, the test error rates of ``\alg + Trim-mean'' and ``\alg + Median'' under different attacks match those of FedAvg in the absence of any attacks. In contrast, existing FL methods show vulnerability; for example, the testing error rate of ``GAS + Median'' reaches 0.28 under the strong Trim attack.

\myparatight{Results of various defenses with subset client selection per training round}By default, we assume full client participation in each round of FL training. Here, we examine a setup where only a subset of clients joins each round. Specifically, the server randomly selects 30\% of the clients each round to receive the current global model. In this configuration, \alg still generates synthetic updates amounting to half of the received local updates. For instance, if the server receives 30 model updates, it generates an additional 15 synthetic ones. Table~\ref{result_table_subset_select} in Appendix shows defense results for this subset selection scenario. We observe that \alg remains robust against various poisoning attacks, whereas existing FL methods show vulnerabilities. For instance, under the Scaling attack, the testing error rate and attack success rate for ``GAS + Trim-mean'' reach 0.36 and 0.41, respectively.

\myparatight{Transferability of \alg}In this part, we demonstrate that our proposed \alg is transferable to other aggregation rules. Specifically, after receiving clients' model updates, the server generates synthetic updates as outlined in Section~\ref{our_method}. Rather than applying the Trimmed-mean or Median aggregation rules, the server instead uses the aggregation methods listed in Table~\ref{result_complex}, such as Krum, FoolsGold, or FLAME, to merge these updates. The results, presented in Table~\ref{result_transferability} in Appendix, show that \alg effectively transfers to different aggregation protocols. For example, with FLAME method, ``\alg + FLAME'' achieves a testing error rate of 0.10 under the Trim attack, whereas FLAME alone results in a 0.17 error rate (refer to Table~\ref{result_complex}).

\myparatight{Scalability of \alg}To illustrate the scalability of our proposed \alg method, we conducted experiments on a production FL system~\cite{shejwalkar2022back,bonawitz2019towards,lai2022fedscale}. Following the setup in~\cite{shejwalkar2022back}, we assume a total of 1,000 clients, with 20\% being malicious. In each training round, the server randomly selects 30\% of clients to participate in the training process. 
The results, displayed in Table~\ref{result_table_production_fl}, show that both ``\alg + Trim-mean'' and ``\alg + Median'' under the Trim attack achieve testing error rates that align with FedAvg’s performance in the absence of an attack, confirming the scalability of \alg.

\begin{table*}[htbp]
  \centering
  \addtolength{\tabcolsep}{-1.985pt}
  \caption{Results of different FL methods, with each client possessing only four labels.}
          % \vspace{-0.06in}
    \begin{tabular}{l|cccccccc}
    \hline
    Aggregation rule & No attack   & LF attack    & Gaussian attack & Trim attack & Krum attack & Min-Max attack & Min-Sum attack & \multicolumn{1}{c}{Scaling attack} \\
     \hline
    FedAvg  &  0.06   & 0.27  & 0.09  &  0.13 &  0.06   &  0.11 &   0.09    &  0.07 / 0.15 \\
    \hline
    Trim-mean & 0.08  &0.31 &0.08  & 0.42 & 0.19   &  0.12   &     0.12  &  0.08 / 0.05 \\
    GAS + Trim-mean &   0.06  &  0.06  &  0.09   & 0.31  & 0.12   &  0.11  &   0.11    & 0.09 / 0.29 \\
    Gaussian + Trim-mean & 0.06   & 0.26   &  0.91    & 0.91&     0.06   & 0.08 &   0.08 & 0.91 / 1.00 \\
    \rowcolor{greyL}
    \alg + Trim-mean & 0.06     & 0.08  &  0.06  &  0.09 &  0.07   &  0.07   &  0.06 & 0.06 / 0.03 \\
     \hline
    Median & 0.07  & 0.32  &   0.33   & 0.58  & 0.25 & 0.19  &  0.20     & 0.09 / 0.05 \\
    GAS + Median & 0.07   &  0.07  &  0.11  & 0.35  & 0.12  & 0.11    &  0.11     &  0.15 / 0.23\\
    Gaussian + Median &  0.06  & 0.90 & 0.89   &  0.89  &  0.07 &   0.17  &  0.18     & 0.91 / 1.00 \\
    \rowcolor{greyL}
    \alg + Median & 0.06 &   0.10   & 0.06  & 0.11   & 0.11  & 0.07      &  0.08     &  0.06 / 0.02 \\
     \hline
    \end{tabular}%
     \label{result_table_non_iid}%
\end{table*}%

%%%%%%%%%%%%%%%%%%%%%%%%%%%%%%%%%%%%%%

\myparatight{Performance of ``\alg + Trim-mean'' when the number of malicious clients is unknown}
Note that in the Trim-mean aggregation method, the server initially excludes the largest $c$ and smallest $c$ values for each dimension before calculating the average of the remaining values. Following previous studies~\cite{blanchard2017machine,yin2018byzantine,liu2023byzantine}, we assume that the trim parameter $c$ equals the total number of malicious clients $f$. However, in practical scenarios, the server may not have exact knowledge of the number of malicious clients, and the attacker only knows that $f$ is bounded by $c$, i.e., $c>f$. Table~\ref{result_malicious_number_unknow} presents the results of our proposed ``\alg + Trim-mean'' approach when the server lacks precise information about the number of malicious clients. Specifically, the trim parameter is set to $c=30$, meaning the server excludes the largest 30 and smallest 30 values per dimension and averages the remaining 40 values (out of 100 clients in total). ``$f=10$'' indicates that there are actually 10 malicious clients. Note that in both cases where ``$f=10$'' and ``$f=20$'', the server still generates $\frac{n}{2}=50$ synthetic updates. 
``Estimate'' refers to the approach in which the server approximates the number of malicious clients in the system. Specifically, in each round, the server calculates the pairwise cosine distances between each pair of client model updates, then applies the K-means~\cite{hartigan1979algorithm} clustering algorithm to divide all client model updates into two clusters. Based on the assumption that the majority of clients are benign, the cluster with fewer clients is considered malicious, and the trim parameter, representing the estimated number of malicious clients, is set to the size of this cluster.
The results in Table~\ref{result_malicious_number_unknow} demonstrate that our proposed ``\alg + Trim-mean'' remains effective even when the actual number of malicious clients is within the bounds defined by the trim parameter or the server approximates the count of malicious clients.

% \begin{figure}[!t]
% 	\centering
% 	\includegraphics[scale = 0.4]{figs/computation_cost.pdf}
% 	\caption{Computation cost of different FL methods.}
% 	\label{fig_computation_cost}
% 		\vspace{-.15in}
% \end{figure}

\begin{figure}[!t]
	\centering
	\includegraphics[scale = 0.46]{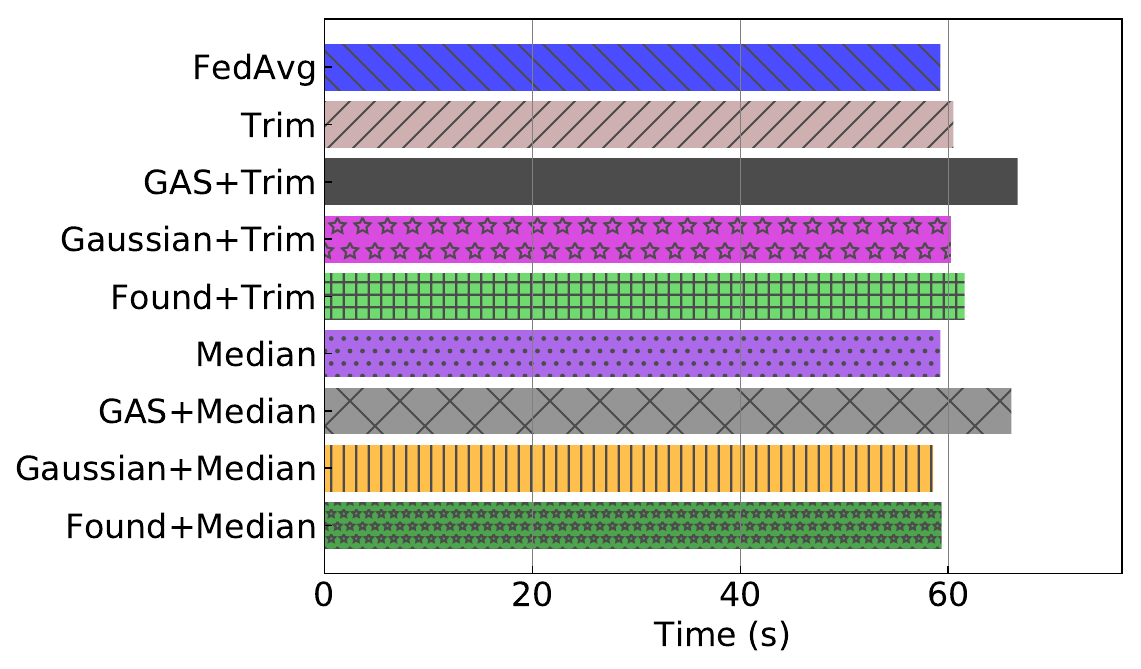}
	\caption{Computation cost of different FL methods.}
	\label{fig_computation_cost}
		% \vspace{-.15in}
\end{figure}

\begin{table*}[htbp]
	\centering
       % \small
       % \footnotesize
       \scriptsize
	\addtolength{\tabcolsep}{-2.3pt}
	\caption{Results of various defenses against DBA and Neurotoxin attacks.}
         % \vspace{-0.06in}
	  \subfloat[MNIST dataset.]
	{
	\begin{tabular}{l|ccc}
		\hline
		Aggregation rule &  DBA  & Neurotoxin \\
		\hline
		FedAvg  & 0.48 / 0.65 & 0.76 / 0.73 \\
		\hline
		Trim-mean & 0.06 / 0.02 & 0.09 / 0.02 \\
		GAS + Trim-mean & 0.27 / 0.19 & 0.43 / 0.54 \\
		Gaussian + Trim-mean  & 0.91 / 1.00 & 0.91 / 1.00 \\
		\rowcolor{greyL}
		\alg + Trim-mean & 0.05 / 0.02 & 0.05 / 0.01 \\
		\hline
		Median & 0.05 / 0.02 & 0.05 / 0.02 \\
		GAS + Median  & 0.43 / 0.51 & 0.76 / 0.82 \\
		Gaussian + Median  & 0.91 / 1.00 & 0.91 / 1.00 \\
		\rowcolor{greyL}
		\alg + Median  & 0.05 / 0.02 & 0.06 / 0.02 \\
		\hline
	\end{tabular}%
	}
 % \hspace{0.02in}
	  \subfloat[Fashion-MNIST dataset.]
	{
	\begin{tabular}{l|ccc}
		\hline
		Aggregation rule &  DBA  & Neurotoxin \\
		\hline
		FedAvg  & 0.76 / 0.72  & 0.71 / 0.79 \\
		\hline
		Trim-mean & 0.29 / 0.06 & 0.30 / 0.05 \\
		GAS + Trim-mean & 0.53 / 0.62 & 0.60 / 0.71 \\
		Gaussian + Trim-mean  & 0.90 / 1.00 & 0.90 / 1.00 \\
		\rowcolor{greyL}
		\alg + Trim-mean & 0.22 / 0.02 & 0.22 / 0.02 \\
		\hline
		Median & 0.27 / 0.03  & 0.31 / 0.03  \\
		GAS + Median   & 0.60 / 0.61 & 0.78 / 0.83  \\
		Gaussian + Median  & 0.90 / 1.00 &  0.90 / 1.00 \\
		\rowcolor{greyL}
		\alg + Median  & 0.23 / 0.03 & 0.22 / 0.02 \\
		\hline
	\end{tabular}%
	}
 % \hspace{0.02in}
	  \subfloat[HAR dataset.]
	{
	\begin{tabular}{l|ccc}
		\hline
		Aggregation rule &  DBA  & Neurotoxin \\
		\hline
		FedAvg  & 0.05 / 0.91 & 0.05 / 0.89 \\
		\hline
		Trim-mean & 0.07 / 0.02 & 0.07 / 0.02 \\
		GAS + Trim-mean & 0.14 / 0.68  & 0.07 / 0.84 \\
		Gaussian + Trim-mean  & 0.59 / 0.02 & 0.75 / 0.16 \\
		\rowcolor{greyL}
		\alg + Trim-mean & 0.05 / 0.01  & 0.05 / 0.01 \\
		\hline
		Median & 0.07 / 0.02 & 0.07 / 0.02 \\
		GAS + Median  & 0.07 / 0.75  & 0.11 / 0.82 \\
		Gaussian + Median  & 0.07 / 0.03 & 0.08 / 0.10 \\
		\rowcolor{greyL}
		\alg + Median  & 0.05 / 0.01 & 0.05 / 0.01 \\
		\hline
	\end{tabular}%
	}
	\\
		  \subfloat[Purchase dataset.]
	{
	\begin{tabular}{l|ccc}
		\hline
		Aggregation rule &  DBA  & Neurotoxin \\
		\hline
		FedAvg  & 0.99 / 0.31  & 0.99 / 0.35 \\
		\hline
		Trim-mean & 0.21 / 0.04  & 0.24 / 0.02 \\
		GAS + Trim-mean & 0.25 / 0.09 & 0.33 / 0.06  \\
		Gaussian + Trim-mean  & 0.99 / 1.00 & 0.99 / 1.00 \\
		\rowcolor{greyL}
		\alg + Trim-mean & 0.20 / 0.02  & 0.18 / 0.01 \\
		\hline
		Median & 0.23 / 0.01 & 0.26 / 0.03 \\
		GAS + Median  & 0.22 / 0.02 & 0.29 / 0.18 \\
		Gaussian + Median  & 0.99 / 1.00 & 0.99 / 1.00 \\
		\rowcolor{greyL}
		\alg + Median  & 0.21 / 0.03 & 0.23 / 0.03 \\
		\hline
	\end{tabular}%
	}
 % \hspace{0.02in}
	\subfloat[CelebA dataset.]
	{
	\begin{tabular}{l|ccc}
		\hline
		Aggregation rule &  DBA  & Neurotoxin \\
		\hline
		FedAvg  & 0.41 / 0.05 & 0.46 / 0.05 \\
		\hline
		Trim-mean & 0.49 / 0.11 & 0.35 / 0.03 \\
		GAS + Trim-mean & 0.36 / 0.04 & 0.36 / 0.03  \\
		Gaussian + Trim-mean   & 0.52 / 0.17 & 0.64 / 0.20 \\
		\rowcolor{greyL}
		\alg + Trim-mean & 0.24 / 0.02  & 0.24 / 0.02 \\
		\hline
		Median & 0.28 / 0.08  & 0.39 / 0.12 \\
		GAS + Median  & 0.32 / 0.05 & 0.31 / 0.09 \\
		Gaussian + Median  & 0.43 / 0.02 & 0.49 / 0.04 \\
		\rowcolor{greyL}
		\alg + Median  & 0.24 / 0.01 & 0.25 / 0.02 \\
		\hline
	\end{tabular}%
	}
 % \hspace{0.02in}
	\subfloat[CIFAR-10 dataset.]
	{
	\begin{tabular}{l|ccc}
		\hline
		Aggregation rule &  DBA  & Neurotoxin \\
		\hline
		FedAvg  & 0.90 / 1.00 & 0.90 / 1.00 \\
		\hline
		Trim-mean & 0.28 / 0.96 & 0.26 / 0.99 \\
		GAS + Trim-mean & 0.25 / 0.80 & 0.31 / 0.92 \\
		Gaussian + Trim-mean  & 0.99 / 1.00 & 0.98 / 1.00 \\
		\rowcolor{greyL}
		\alg + Trim-mean & 0.22 / 0.01 & 0.23 / 0.02 \\
		\hline
		Median & 0.31 / 0.84 & 0.37 / 0.91 \\
		GAS + Median  & 0.45 / 0.96 & 0.32 / 0.99 \\
		Gaussian + Median  & 0.99 / 1.00 & 0.99 / 1.00 \\
		\rowcolor{greyL}
		\alg + Median  & 0.24 / 0.02 & 0.26 / 0.03 \\
		\hline
	\end{tabular}%
	}
	\label{result_more_backdoor}%
 		% \vspace{-.15in}
\end{table*}%

\myparatight{Computation cost of different FL methods}Fig.~\ref{fig_computation_cost} illustrates the computational costs of various FL methods during 2,000 rounds of training on the MNIST dataset. The computation cost represents the time required by the server to aggregate model updates over these rounds. It is important to note that for our proposed method, the computation cost also includes the time taken to generate synthetic updates. In Fig.~\ref{fig_computation_cost}, ``Trim'', ``GAS+Trim'', ``Gaussian+Trim'', ``Found+Trim'', and  ``Found+Median'' refer to the methods ``Trim-mean'', ``GAS + Trim-mean'', ``Gaussian + Trim-mean'', ``\alg + Trim-mean'', and  ``\alg + Median'', respectively. 
% From Fig.~\ref{fig_computation_cost}, we observe that our proposed ``\alg + Trim-mean'' and ``\alg + Median''  methods do not significantly increase computation costs compared to FedAvg.
%
The additional computational overhead introduced by \alg primarily stems from the generation of synthetic updates, which are designed to enhance robustness. Although this process requires extra computations compared to standard FL approaches, the overhead remains manageable and does not compromise system robustness or accuracy. 
The synthetic update generation process has been fine-tuned to operate efficiently within each round, ensuring that the overall latency does not diverge significantly from baseline methods. Additionally, as shown in Fig.~\ref{fig_computation_cost}, the method incorporating \alg with a robust aggregator (Trim-mean or Median) exhibits slightly higher computational overhead compared to FedAvg, underscoring the effectiveness of our approach in balancing robustness and efficiency.

% !TEX root = main.tex

\section{Discussion and Limitations} 
\label{sec:discussion_limitation}

\myparatight{More extreme Non-IID distribution}In the previous section, we demonstrated the effectiveness of our proposed method in safeguarding FL systems against highly heterogeneous local training data among clients. Here, we explore a more extreme scenario where each client possesses only a few distinct labels. Specifically, we consider a situation where each client has only four different labels in their training data. For instance, one client may have data labeled from one to four, while another client may have data labeled from five to eight. The results of various FL methods under different attack conditions in this extreme setting are presented in Table~\ref{result_table_non_iid}. From Table~\ref{result_table_non_iid}, it is evident that \alg can effectively mitigate poisoning attacks in FL systems even under such challenging conditions.

\myparatight{More targeted attacks}Here, we demonstrate the robustness of various defense mechanisms against two more targeted attacks: the DBA attack~\cite{xie2019dba} and the Neurotoxin attack~\cite{zhang2022neurotoxin}. Table~\ref{result_more_backdoor} presents the results across six datasets, where ``DBA'' and ``Neurotoxin'' represent the respective attacks. Our \alg method shows low testing error rates and attack success rates under these attacks. In contrast, existing aggregation rules remain vulnerable to poisoning attacks; for example, the attack success rate of Trim-mean reaches 0.96 on the CIFAR-10 dataset under the DBA attack.

\myparatight{More sophisticated and adaptive attacks}In Section~\ref{sec:exp}, we show the capability of our proposed \alg framework to effectively mitigate seven distinct poisoning attacks. In this section, we further examine three additional sophisticated attacks, which include two adaptive attack strategies.

\myparatight{a) MPAF attack~\cite{cao2022mpaf}}MPAF is an untargeted attack where the attacker steers the current global model towards an attacker-chosen model during each training round.

\myparatight{b) Adaptive attack I~\cite{baruch2019little}}In this untargeted attack, the attacker introduces small perturbations to benign local model updates to hinder the convergence of the global model. As shown in~\cite{baruch2019little}, this method allows the attacker to effectively compromise the final learnt global model while evading detection.

\myparatight{c) Adaptive attack II~\cite{baruch2019little}}This attack is both targeted and adaptive, involving the insertion of backdoor triggers into local training data by the attacker on a malicious client. However, unlike scaling malicious local model updates with a fixed factor, the attacker dynamically determines the scaling factor through an optimization process.

Table~\ref{result_adaptive} presents the results of various FL methods under sophisticated attack scenarios. 
``MPAF'', ``Adaptive I'', and ``Adaptive II'' represent the ``MPAF attack'', ``Adaptive attack I'', and ``Adaptive attack II'' respectively.
From the table, it is evident that despite the attacker's efforts to evade detection, our proposed \alg method effectively mitigates these sophisticated attacks.

\begin{table}[t]
  \centering
  % \addtolength{\tabcolsep}{-1.85pt}
  \caption{Results of different FL methods under more sophisticated and adaptive attacks.}
    \begin{tabular}{l|ccc}
    \hline
    Aggregation rule & MPAF &  Adaptive I  & Adaptive II \\
     \hline
    FedAvg  & 0.90  & 0.90 &  0.90 / 1.00 \\
    \hline
    Trim-mean  &  0.23 & 0.21 &  0.19 / 0.06 \\
     GAS + Trim-mean  &  0.13 & 0.28 &  0.60 / 0.62 \\
    Gaussian + Trim-mean  & 0.90  & 0.90 &  0.91 / 1.00 \\
    \rowcolor{greyL}
    \alg + Trim-mean  & 0.05  & 0.06 & 0.06 / 0.02  \\
     \hline
    Median  & 0.29  & 0.17 &  0.13 / 0.02 \\
    GAS + Median  &  0.06 & 0.19 & 0.71 / 0.54  \\
    Gaussian + Median  & 0.91  & 0.91 & 0.91 / 1.00  \\
    \rowcolor{greyL}
    \alg + Median  & 0.06  & 0.06 & 0.06 / 0.03  \\
     \hline
    \end{tabular}%
     \label{result_adaptive}%
\end{table}%

\myparatight{Further discussion on the threat model for the ratio of malicious clients}From Fig.~\ref{fig_attack_size}, we observe that our proposed \alg framework can tolerate up to 45\% of malicious clients. However, when 50\% of clients are malicious, the testing error rates of \alg become significantly higher; for example, under the Krum attack, error rates rise to 0.64 for ``\alg + Trim-mean'' and 0.27 for ``\alg + Median''. Nonetheless, compromising such a large fraction of malicious clients is impractical, as demonstrated in~\cite{shejwalkar2022back}. This is because achieving such high numbers of malicious clients in a real-world FL setup is unlikely due to the distributed and decentralized nature of the system, making it challenging to mobilize or control such a large fraction of participants for malicious purposes.

\myparatight{Potential challenges introduced by \alg}Our proposed methodology incorporates synthetic updates to mitigate the influence of potentially malicious updates within the system. By employing techniques such as Trimmed-Mean or Median aggregation, we effectively reduce the weight of outlier contributions, thereby minimizing their impact and preventing the model from overfitting. 
Through extensive experiments, we demonstrate that this approach does not degrade performance metrics, such as testing accuracy, serving as evidence that our method is not prone to overfitting. 
Although bias is not the focus of this paper, we can pair our \alg with bias reduction techniques, such as regularization and fairness-constrained optimization methods~\cite{abay2020mitigating,cui2021addressing,guo2023fedbr}. 
By combining our synthetic update strategy with these bias reduction techniques, we can create a more robust and fair model that not only performs well but also addresses potential biases.

\myparatight{Limitations of \alg}In this study, we show that creating entirely new robust aggregation protocols may not be necessary to secure FL systems effectively. Rather, by strengthening the robustness of existing Byzantine-resistant foundational aggregation methods, such as Trimmed-mean or Median, we can achieve substantial resilience against poisoning attacks. Nevertheless, our proposed \alg framework has certain limitations. Firstly, it is restricted to coordinate-wise aggregation rules, limiting its compatibility with other aggregation approaches. Secondly, \alg introduces a slightly higher computational overhead compared to the commonly used FedAvg approach. 

% !TEX root = main.tex

%\vspace{-.1in}
% \newpage

\section{Conclusion} \label{sec:conclusion}

In this work, we introduced a new approach, referred to as \alg, aimed at countering poisoning attacks in FL systems. Rather than designing intricate new Byzantine-robust aggregation protocols, our goal is to bolster the resilience of FL systems using established Byzantine-robust foundational aggregation protocols. In our proposed framework, the server takes a proactive stance by generating synthetic updates upon receiving local model updates from clients. Subsequently, the server employs existing Byzantine-robust foundational aggregation protocols like Trimmed-mean or Median to merge the local model updates from clients with the generated synthetic updates. We demonstrated the convergence performance of our framework under poisoning attacks and conducted extensive experiments across diverse scenarios to validate the effectiveness of our proposed techniques. In the future, we intend to expand our approach to incorporate other non-coordinate wise aggregation protocols such as Krum.

\section*{Acknowledgement}
We thank the anonymous reviewers for their constructive comments.

\bibliographystyle{IEEEtranS}
\bibliography{refs}

%%%% appendix.tex starts here %%%%

\appendix

\subsection{Proof of Theorem~\ref{Theorem1}} 
\label{sec:appendix_1}

Following~\cite{cao2020fltrust}, we adopt a slight notation abuse in our proof, using $\bm{g}_i^t$ to denote the gradient of client $i$ in training round $t$.
Prior to proving our main theoretical results, we first present several helpful lemmas.

\begin{lem}
\label{trim_lem}
Suppose that Assumptions~\ref{assumption_1}-\ref{assumption_3} and Assumption~\ref{assumption_5} are satisfied, and the server employs  the Trimmed-mean aggregation rule to merge  the synthetic model updates and model updates from clients.
At training round $t$, there exists a probability of at least $1-\frac{4d}{(1+ (n+m) \lambda \varpi Q)^d}$ such that the following holds:
\begin{align}
\left\|  \bm{g}(\bm{\theta}^t) -  \nabla L(\bm{\theta}^t)   \right\| \le B_1,
 \end{align}
where $\bm{g}(\bm{\theta}^t)$ is the global model update, $B_1=\mathcal{O}((\frac{\rho c d}{\upsilon (n+m) \sqrt{Q}} + \frac{\rho d}{\upsilon \sqrt{(n+m) Q}}) \sqrt{ \log((n+m) \lambda \varpi Q)} ) $.
\begin{proof}
 The proof proceeds similarly to Theorem 11 in~\cite{yin2018byzantine}, and we omit it here for conciseness.
\end{proof}
 
\end{lem}

\begin{lem}
\label{general_lem}
Suppose Assumptions~\ref{assumption_1}-\ref{assumption_2} hold. If the learning rate $\alpha$ 
used by the clients satisfies $\alpha=\frac{1}{\lambda}$, then in training round $t$, we have:
\begin{align}
\left\| \bm{\theta}^{t} - \alpha \nabla L(\bm{\theta}^{t}) - \bm{\theta} ^* \right\| \le  \left(1-\frac{\mu}{\mu+\lambda}\right)  \left\| \bm{\theta}^ {t}  - \bm{\theta} ^* \right\|, \nonumber
\end{align}
where $\bm{\theta}^*$ is the optimal model under no attack.

\begin{proof}
We start by analyzing the squared norm:
\begin{align}
\left\| \bm{\theta}^{t} - \alpha \nabla L(\bm{\theta}^{t}) - \bm{\theta} ^* \right\|^2 
&= \left\| \bm{\theta}^{t} - \bm{\theta} ^* \right\|^2 + \alpha^2 \left\| \nabla L(\bm{\theta}^{t})  \right\|^2  \nonumber \\
& \quad -2 \alpha \left\langle \bm{\theta}^{t}- \bm{\theta} ^*, \nabla L(\bm{\theta}^{t}) \right\rangle.
\end{align}

According to~\cite{bubeck2015convex}, for any $\bm{\theta}_1, \bm{\theta}_2 \in \Theta$, we have:
\begin{align}
\frac{\mu \lambda}{\mu+\lambda} \left\| \bm{\theta}_1 - \bm{\theta}_2 \right\|^2 + \frac{1}{\mu+\lambda} \left\| \nabla L(\bm{\theta}_1) - \nabla L(\bm{\theta}_2) \right\|^2  \nonumber \\
\le \left\langle \nabla L(\bm{\theta}_1) - \nabla  L(\bm{\theta}_2), \bm{\theta}_1 - \bm{\theta}_2 \right\rangle.
\end{align}

Setting $\bm{\theta}_1 = \bm{\theta}^t$ and $\bm{\theta}_2 = \bm{\theta}^*$, and noting $\nabla L(\bm{\theta}^*)=0$, we obtain:
\begin{align}
\frac{\mu \lambda}{\mu+\lambda} \left\| \bm{\theta}^t - \bm{\theta}^* \right\|^2 +
\frac{1}{\mu+\lambda} \left\| \nabla L(\bm{\theta}^t)  \right\|^2  \nonumber \\
\le 
\left\langle \bm{\theta}^t - \bm{\theta}^*, \nabla L(\bm{\theta}^t) \right\rangle.
\end{align}

Furthermore, with $\alpha=\frac{1}{\lambda}$, we derive:
\begin{align}
\left\| \bm{\theta}^{t} - \alpha \nabla L(\bm{\theta}^{t}) - \bm{\theta} ^* \right\|^2 
\le \left(1-\frac{2\mu}{\mu+\lambda}\right) \left\| \bm{\theta}^{t} - \bm{\theta}^* \right\|^2.
\end{align}

Given $\mu \le \lambda$, it follows that:
\begin{align}
\left\| \bm{\theta}^{t} - \alpha \nabla L(\bm{\theta}^{t}) - \bm{\theta} ^* \right\| \le \left(1-\frac{\mu}{\mu+\lambda}\right) \left\| \bm{\theta}^ {t}  - \bm{\theta}^* \right\|,
\end{align}
completing the proof.
\end{proof}

\end{lem}

\myparatight{Proof of Theorem~\ref{Theorem1}} Given the lemmas presented earlier, we proceed to establish Theorem~\ref{Theorem1}. In the $t$-th global training round, the following holds:
\begin{align}
\label{theorem_one_first}
\| \bm{\theta}^{t +1} -\bm{\theta}^* \| = \| \bm{\theta}^t - \alpha \nabla \bm{g}(\bm{\theta}^t) - \bm{\theta}^* \|.
\end{align}

Applying the triangle inequality to the gradient terms, the right-hand side of Eq.~(\ref{theorem_one_first}) satisfies:
\begin{align}
\label{theorem_one_two}
\leq \| \bm{\theta}^t - \alpha \nabla L(\bm{\theta}^t) - \bm{\theta}^* \| + \alpha \| \nabla \bm{g}(\bm{\theta}^t) - \nabla L(\bm{\theta}^t) \|.
\end{align}

Based on the above Lemmas~\ref{trim_lem}-\ref{general_lem}, this can be simplified to:
\begin{align}
\label{theorem_one_three}
\leq (1-\frac{\mu}{\mu+\lambda}) \| \bm{\theta}^t - \bm{\theta}^* \| + \frac{B_1}{\lambda},
\end{align}
where $B_1$ is defined as:
\begin{align}
B_1 &=\mathcal{O}((\frac{\rho c d}{\upsilon (n+m) \sqrt{Q}}  \nonumber \\ 
& \quad + \frac{\rho d}{\upsilon \sqrt{(n+m) Q}}) \sqrt{ \log((n+m) \lambda \varpi Q)}).
\end{align}

Applying the condition \(\alpha=\frac{1}{\lambda}\), and since \(\mu \le \lambda \), then after $T$ global training rounds, one can further have the following:
\begin{align}
 \left\| \bm{\theta}^T - \bm{\theta}^* \right\| \le 
  (1-\frac{\mu}{\mu+\lambda})^T   \left\| \bm{\theta}^0 - \bm{\theta}^* \right\| + \frac{2 B_1}{\mu}.
\end{align}

This concludes the convergence proof of the global model under attack to the optimal point under the given conditions and assumptions.

\begin{lem}
\label{median_lem}
% Under the assumptions that Assumptions~\ref{assumption_1}-\ref{assumption_4} hold true and the client's learning rate is set as $\alpha=\frac{1}{\lambda}$, our proposed framework, denoted as \alg, uses the Median aggregation rule to merge synthetic updates and updates contributed by clients. Assuming $\upsilon > 0$, if the proportion $\beta$ of malicious clients satisfies $\beta + \epsilon \le \frac{1}{2} - \upsilon$, then after $T$ rounds of global training, the probability of achieving the following outcome is guaranteed to be at least $1 - \frac{4d}{(1 + n \lambda \varpi Q)^d}$:
%
%
Assuming Assumptions~\ref{assumption_1} to \ref{assumption_4} hold, our proposed framework, \alg, employs the Median aggregation rule to integrate synthetic updates and client-contributed updates. Given $\upsilon > 0$, if $\beta=\frac{f}{n+m}$ satisfies $\beta + \epsilon \le \frac{1}{2} - \upsilon$, then after $T$ rounds of global training, the probability of achieving the following outcome is guaranteed to be at least $1 - \frac{4d}{(1 + (n+m) \lambda \varpi Q)^d}$:
\begin{align}
\left\|  \bm{g}(\bm{\theta}^t) -  \nabla L(\bm{\theta}^t)   \right\| \le B_2,
 \end{align}
where $\bm{g}(\bm{\theta}^t)$ is the global model update, 
$\epsilon$ is defined as
$\epsilon=\frac{0.4748  \zeta}{\sqrt{Q}} + \sqrt{\frac{d\log(1+(n+m) \lambda \varpi Q)}{(n+m)(1-\beta)}} $,
$B_2=\frac{2\sqrt{2}}{(n+m) Q} + \frac{2 \sqrt{\pi} \sigma (\beta + \epsilon) \exp(\frac{1}{2}(\Phi^{-1}(1-\upsilon))^2)}{\sqrt{Q}}$.
\begin{proof}
 The proof follows a similar approach to Theorem 8 in~\cite{yin2018byzantine}, and we omit it here for brevity.
\end{proof}
\end{lem}

\subsection{Proof of Theorem~\ref{Theorem2}} 
\label{sec:appendix_2}
The proof of Theorem~\ref{Theorem2} follow the same procedure as that of Theorem~\ref{Theorem1}. 
% which we omit here for brevity.
%
The only difference is that we simplify Eq.~(\ref{theorem_one_two}) using Lemma~\ref{general_lem} and Lemma~\ref{median_lem} rather than Lemmas~\ref{trim_lem}-\ref{general_lem}. For brevity, we omit the full proof.

%  \begin{table}[htb!]
% 	%\centering
% 	% \small
% 	\caption{The default CNN architecture.}
% 	\centering
% 	% \vspace{1mm}
% 	% 	\small
% 	% \footnotesize 
% 	%	\scriptsize
% 	\begin{tabular}{|c|c|} \hline 
% 		{Layer} & {Size} \\ \hline
% 		{Input} & { $28\times28\times1$}\\ \hline
% 		{Convolution + ReLU} & { $3\times3\times30$}\\ \hline
% 		{Max Pooling} & { $2\times2$}\\ \hline
% 		{Convolution + ReLU} & { $3\times3\times50$}\\ \hline
% 		{Max Pooling} & { $2\times2$}\\ \hline
% 		{Fully Connected + ReLU} & {100}\\ \hline
% 		{Softmax} & {10}\\ \hline
% 	\end{tabular}
% 	\label{cnn_arch}
% 	% 		\vspace{-2mm}
% \end{table}

 \begin{table}[htb!]
	%\centering
%	 \small
		% \addtolength{\tabcolsep}{-4.7pt}
	\caption{CNN architectures.}
	\centering
	% \vspace{1mm}
	% 	\small
	% \footnotesize 
	%	\scriptsize
	 \subfloat[The default CNN architecture.]
	{
	\begin{tabular}{|c|c|} \hline 
		{Layer} & {Size} \\ \hline
		{Input} & { $28\times28\times1$}\\ \hline
		{Convolution + ReLU} & { $3\times3\times30$}\\ \hline
		{Max Pooling} & { $2\times2$}\\ \hline
		{Convolution + ReLU} & { $3\times3\times50$}\\ \hline
		{Max Pooling} & { $2\times2$}\\ \hline
		{Fully Connected + ReLU} & {100}\\ \hline
		{Softmax} & {10}\\ \hline
	\end{tabular}
	\label{cnn_arch}
	}
    \\
	 \subfloat[An alternative CNN architecture.]
	{
		\begin{tabular}{|c|c|} \hline 
		{Layer} & {Size} \\ \hline
		{Input} & { $28\times28\times1$}\\ \hline
		{Convolution + ReLU} & { $3\times3\times30$}\\ \hline
		% {Max Pooling} & { $2\times2$}\\ \hline
		% {Convolution + ReLU} & { $3\times3\times50$}\\ \hline
		% {Max Pooling} & { $2\times2$}\\ \hline
		{Fully Connected + ReLU} & {100}\\ \hline
		{Softmax} & {10}\\ \hline
	\end{tabular}
        \label{cnn_arch_diff}
	}
	% 		\vspace{-2mm}
\end{table}

\begin{table*}[htbp]
 % \vspace{-1.06in}
  \centering
  \addtolength{\tabcolsep}{-1.985pt}
  \caption{Results of different FL methods on CelebA and CIFAR-10 datasets. The results of Scaling attack are shown as ``testing error rate / attack success rate''.}
\label{result_all_datasets_CelebA_cifar10}
        \vspace{-0.06in}
       \subfloat[CelebA dataset.]
   {
    \begin{tabular}{l|cccccccc}
    \hline
    Aggregation rule & No attack   & LF attack    & Gaussian attack & Trim attack & Krum attack & Min-Max attack & Min-Sum attack & \multicolumn{1}{c}{Scaling attack} \\
     \hline
    FedAvg  & 0.23  & 0.33 & 0.37  & 0.44 &   0.56 & 0.56 &  0.57 & 0.45 / 0.02 \\
    \hline
    Trim-mean &0.31  & 0.33 & 0.30  & 0.48  & 0.36 & 0.56  &  0.52  & 0.34 / 0.07 \\
    GAS + Trim-mean & 0.33 & 0.34  & 0.34 & 0.48 & 0.53 &  0.48 &  0.51 & 0.36 / 0.03  \\
    Gaussian + Trim-mean &  0.23  &  0.30  & 0.53 &  0.32 & 0.42 &  0.54  & 0.52  & 0.48 / 0.05 \\
    \rowcolor{greyL}
    \alg + Trim-mean & 0.23  &  0.23   &  0.23 &  0.23  &  0.25  &  0.23 &  0.23  & 0.24 / 0.02 \\
     \hline
    Median &  0.31  &  0.32 & 0.31  & 0.46 &0.35 &  0.47 &   0.47 & 0.33 / 0.05 \\
    GAS + Median &   0.32 & 0.35  &  0.35 &  0.48 & 0.53  &  0.48 & 0.51 & 0.35 / 0.03 \\
    Gaussian + Median & 0.25  &  0.25 & 0.53 & 0.41  & 0.44   & 0.56 & 0.57  & 0.49 / 0.04 \\
    \rowcolor{greyL}
    \alg + Median &  0.24     &  0.25  & 0.25  &  0.26   &  0.24  & 0.25  & 0.24  & 0.26 / 0.02 \\
     \hline
    \end{tabular}%
    }
       \\
        \vspace{-0.06in}
       \subfloat[CIFAR-10 dataset.]
   {
    \begin{tabular}{l|cccccccc}
    \hline
    Aggregation rule & No attack   & LF attack    & Gaussian attack & Trim attack & Krum attack & Min-Max attack & Min-Sum attack & \multicolumn{1}{c}{Scaling attack} \\
     \hline
    FedAvg  & 0.22 & 0.30 & 0.72  & 0.83 & 0.23   & 0.23 &  0.23 & 0.90 / 1.00 \\
    \hline
    Trim-mean & 0.25 & 0.32 & 0.26  & 0.79  & 0.26   & 0.25 & 0.25  & 0.28 / 0.96\\
    GAS + Trim-mean & 0.25& 0.29 & 0.80  & 0.90 & 0.25   & 0.25 & 0.26  &  0.37 / 0.93 \\
    Gaussian + Trim-mean & 0.28 & 0.38 & 0.90  & 0.71 & 0.33   & 0.28 & 0.32  & 0.99 / 1.00 \\
    \rowcolor{greyL}
    \alg + Trim-mean &0.22 & 0.23 & 0.22  & 0.25 &  0.22  & 0.23 & 0.23  & 0.22 / 0.02 \\
     \hline
    Median & 0.25 & 0.30  & 0.25  & 0.84 & 0.30   & 0.27& 0.26  & 0.28 / 0.96 \\
    GAS + Median &0.25 & 0.26 & 0.78  & 0.90 & 0.28   & 0.25 & 0.25  & 0.27 / 0.89\\
    Gaussian + Median & 0.32 & 0.72 & 0.85  & 0.76 & 0.80   &0.67 &  0.75 & 0.90 / 1.00 \\
    \rowcolor{greyL}
    \alg + Median & 0.23 & 0.25  & 0.23  & 0.24 &  0.23  & 0.23 & 0.23  & 0.24 / 0.02\\
     \hline
    \end{tabular}%
    \label{result_all_datasets_cifar10}
    }
     % \vspace{-.15in}
\end{table*}%

\begin{table*}[htbp]
  \centering
  \addtolength{\tabcolsep}{-1.985pt}
  \caption{Results of different defense approaches on a distinct CNN architecture.}
    \begin{tabular}{l|cccccccc}
    \hline
    Aggregation rule & No attack   & LF attack    & Gaussian attack & Trim attack & Krum attack & Min-Max attack & Min-Sum attack & \multicolumn{1}{c}{Scaling attack} \\
     \hline
    FedAvg  &  0.08  & 0.08  & 0.12   & 0.23 & 0.09   & 0.17  &  0.18    &  0.47 / 0.66 \\
    \hline
    Trim-mean &  0.10  &  0.11 &  0.10 & 0.24  & 0.10   &  0.20 &  0.25    & 0.12 / 0.02  \\
    GAS + Trim-mean & 0.08   & 0.08  & 0.11  & 0.22  & 0.09   & 0.13  & 0.13     &  0.39 / 0.42 \\
    Gaussian + Trim-mean &  0.08  &  0.91 & 0.91  & 0.91 &  0.08  & 0.13  &   0.12   &  0.91 / 1.00 \\
    \rowcolor{greyL}
    \alg + Trim-mean & 0.08   & 0.08   & 0.08   & 0.08  & 0.08    & 0.08   & 0.08      &  0.08 / 0.01 \\
     \hline
    Median & 0.09   & 0.19  & 0.20  & 0.25 &  0.09  & 0.15  &  0.14    & 0.15 / 0.02  \\
    GAS + Median &  0.08  &  0.08  & 0.13  & 0.28 & 0.08   & 0.13  & 0.13      & 0.55 / 0.61  \\
    Gaussian + Median & 0.08   &  0.16 & 0.31  & 0.75 & 0.08  & 0.11  &  0.10    & 0.89 / 1.00  \\
    \rowcolor{greyL}
    \alg + Median &  0.08  & 0.09   & 0.08   &  0.08 &  0.08   & 0.08   &    0.08   &  0.08 / 0.01 \\
     \hline
    \end{tabular}%
     \label{result_table_diff_arch}%
\end{table*}%

\begin{table*}[htbp]
  \centering
  \addtolength{\tabcolsep}{-1.985pt}
  \caption{Results of different defenses when only a subset of clients are selected in each training round.}
    \begin{tabular}{l|cccccccc}
    \hline
    Aggregation rule & No attack   & LF attack    & Gaussian attack & Trim attack & Krum attack & Min-Max attack & Min-Sum attack & \multicolumn{1}{c}{Scaling attack} \\
     \hline
    FedAvg  & 0.06  &  0.08 & 0.12   &  0.13 & 0.06  &  0.75 &  0.75   & 0.68 / 0.81  \\
    \hline
    Trim-mean &  0.06 & 0.07 & 0.07   & 0.26  &  0.08 & 0.11  &  0.15   & 0.13 / 0.02  \\
    GAS + Trim-mean & 0.06  &  0.06 &  0.12  & 0.22 &  0.06 & 0.09  &  0.09   & 0.36 / 0.41  \\
    Gaussian + Trim-mean  & 0.06  & 0.12 &  0.91  &  0.91   & 0.06  & 0.07  & 0.06   &  0.89 / 0.92 \\
    \rowcolor{greyL}
    \alg + Trim-mean & 0.06  & 0.06 & 0.06   & 0.06  & 0.06  & 0.06  &   0.06  & 0.06 / 0.02  \\
     \hline
    Median & 0.06  & 0.08 & 0.07   & 0.28  & 0.29  & 0.11  &  0.13   & 0.06 / 0.02  \\
    GAS + Median & 0.06  & 0.06 & 0.12    & 0.25   & 0.07  & 0.09  &  0.10   &  0.55 / 0.65 \\
    Gaussian + Median & 0.06  & 0.90 & 0.90   & 0.90  & 0.06  & 0.09   &  0.09    & 0.91 / 1.00  \\
    \rowcolor{greyL}
    \alg + Median & 0.06  & 0.06  &  0.06   & 0.06   & 0.06   & 0.06   &   0.06   & 0.06 / 0.02  \\
     \hline
    \end{tabular}%
     \label{result_table_subset_select}%
\end{table*}%

\begin{table*}[htbp]
  \centering
  \addtolength{\tabcolsep}{-1.985pt}
  \caption{Transferability of \alg.}
    \begin{tabular}{l|cccccccc}
    \hline
    Aggregation rule & No attack   & LF attack    & Gaussian attack & Trim attack & Krum attack & Min-Max attack & Min-Sum attack & \multicolumn{1}{c}{Scaling attack} \\
     \hline
     \alg + Krum  & 0.07 & 0.07 &  0.07 & 0.07  & 0.85 & 0.08  & 0.08 & 0.08 / 0.01 \\
     \alg + FoolsGold  & 0.06 & 0.09 & 0.08 & 0.13  & 0.06 & 0.08 & 0.08 & 0.09 / 0.02 \\
     \alg + FLAME  & 0.07 & 0.08 & 0.08 & 0.10  & 0.07 & 0.07 & 0.07 & 0.07 / 0.02\\
     \hline
    \end{tabular}%
     \label{result_transferability}
\end{table*}%

%  \begin{table}[htb!]
% 	%\centering
% 	% \small
% 	\caption{\textcolor{red}{An alternative CNN architecture.}}
% 	\centering
% 	% \vspace{1mm}
% 	% 	\small
% 	% \footnotesize 
% 	%	\scriptsize
% 	\begin{tabular}{|c|c|} \hline 
% 		{Layer} & {Size} \\ \hline
% 		{Input} & { $28\times28\times1$}\\ \hline
% 		{Convolution + ReLU} & { $3\times3\times30$}\\ \hline
% 		% {Max Pooling} & { $2\times2$}\\ \hline
% 		% {Convolution + ReLU} & { $3\times3\times50$}\\ \hline
% 		% {Max Pooling} & { $2\times2$}\\ \hline
% 		{Fully Connected + ReLU} & {100}\\ \hline
% 		{Softmax} & {10}\\ \hline
% 	\end{tabular}
% 	\label{cnn_arch_diff}
% 	% 		\vspace{-2mm}
% \end{table}

\end{document}